\documentclass{scrartcl}

\usepackage{graphicx}
\usepackage{algorithm}
\usepackage{algorithmic}
\usepackage{mathptmx}
\usepackage{amsmath, amsthm, bm}
\usepackage{amsfonts, amssymb}
\usepackage{subfig}
\usepackage{float}
\usepackage{url}
\usepackage{verbatim}
\usepackage{textcomp}

\theoremstyle{plain}
\newtheorem{theorem}{Theorem}[section]

\DeclareGraphicsExtensions{.jpg,.mps,.pdf,.png} \graphicspath{{tex/images/}, {images/}}

\newcommand{\algolong}{MinimizAtion of Cbs Heuristic}

\newcommand{\algo}{MaCH}
\newcommand{\gvns}{GVNS}
\newcommand{\cbs}{CBS}
\newcommand{\cbsp}{CBSP}

\title{Discovering the structure of complex networks by minimizing cyclic 
bandwidth sum}

\author{Ronan Hamon \and Pierre Borgnat \and Patrick Flandrin \and 
C\'eline Robardet}

\date{}
\begin{document}

\maketitle

\begin{abstract}
{Getting a labeling of vertices close to the structure of the graph has been 
proved to be of interest in many applications e.g., to follow smooth signals 
indexed by the vertices of the network. This question can be related to a graph 
labeling problem known as the cyclic bandwidth sum problem. It consists in 
finding a labeling of the vertices of an undirected and unweighted graph with 
distinct integers such that the sum of (cyclic) difference of labels of adjacent 
vertices is minimized. Although theoretical results exist that give optimal 
value of cyclic bandwidth sum for standard graphs, there are neither results in 
the general case, nor explicit methods to reach this optimal result. In addition 
to this lack of theoretical knowledge, only a few methods have been proposed to 
approximately solve this problem. In this paper, we introduce a new heuristic to 
find an approximate solution for the cyclic bandwidth sum problem, by following 
the structure of the graph. The heuristic is a two-step algorithm: the first 
step consists of traversing the graph to find a set of paths which follow the 
structure of the graph, using a similarity criterion based on the Jaccard index 
to jump from one vertex to the next one. The second step is the merging of all 
obtained paths, based on a greedy approach that extends a partial solution by 
inserting a new path at the position that minimizes the cyclic bandwidth sum. 
The effectiveness of the proposed heuristic, both in terms of performance and 
time execution, is shown through experiments on graphs whose optimal value of 
CBS is known as well as on real-world networks, where the consistence between 
labeling and topology is highlighted. An extension to weighted graphs is also 
proposed.}

{cyclic bandwidth sum problem, graph labeling, complex networks, vertex 
labeling, graph structure, graph topology}
		
\end{abstract}

\makeatletter{}\section{Introduction}
\label{sec:introduction}

\subsection{Problem statement}

In many applications, the structure of a complex network gives insights into the 
understanding of the underlying relationships between the vertices: It is 
advantageous to consider the vertices in a consistent order according to the 
topology. A striking example of this is the huge amount of works about the 
detection of communities in a network \cite{Fortunato2010}: finding groups of 
vertices highly connected between them is for instance a powerful tool to 
explain the structure of social networks and to characterize them. The presence 
of communities is only one type of organization encountered in networks, among a 
much large diversity of structures: in many cases, the topology of the network 
is unknown and cannot be fully characterized explicitly. In this situation, it 
is nevertheless beneficial to have a vertex ordering consistent with the 
structure of the network. We can point out for example those related to 
distributed inference over networks \cite{Kar2013}, diffusion \cite{Chamley2013} 
or visualization of networks \cite{Bertrand2013}.

\begin{figure}[!ht]
  \centering
  
    \subfloat[]{
		\includegraphics[width=0.21\columnwidth]{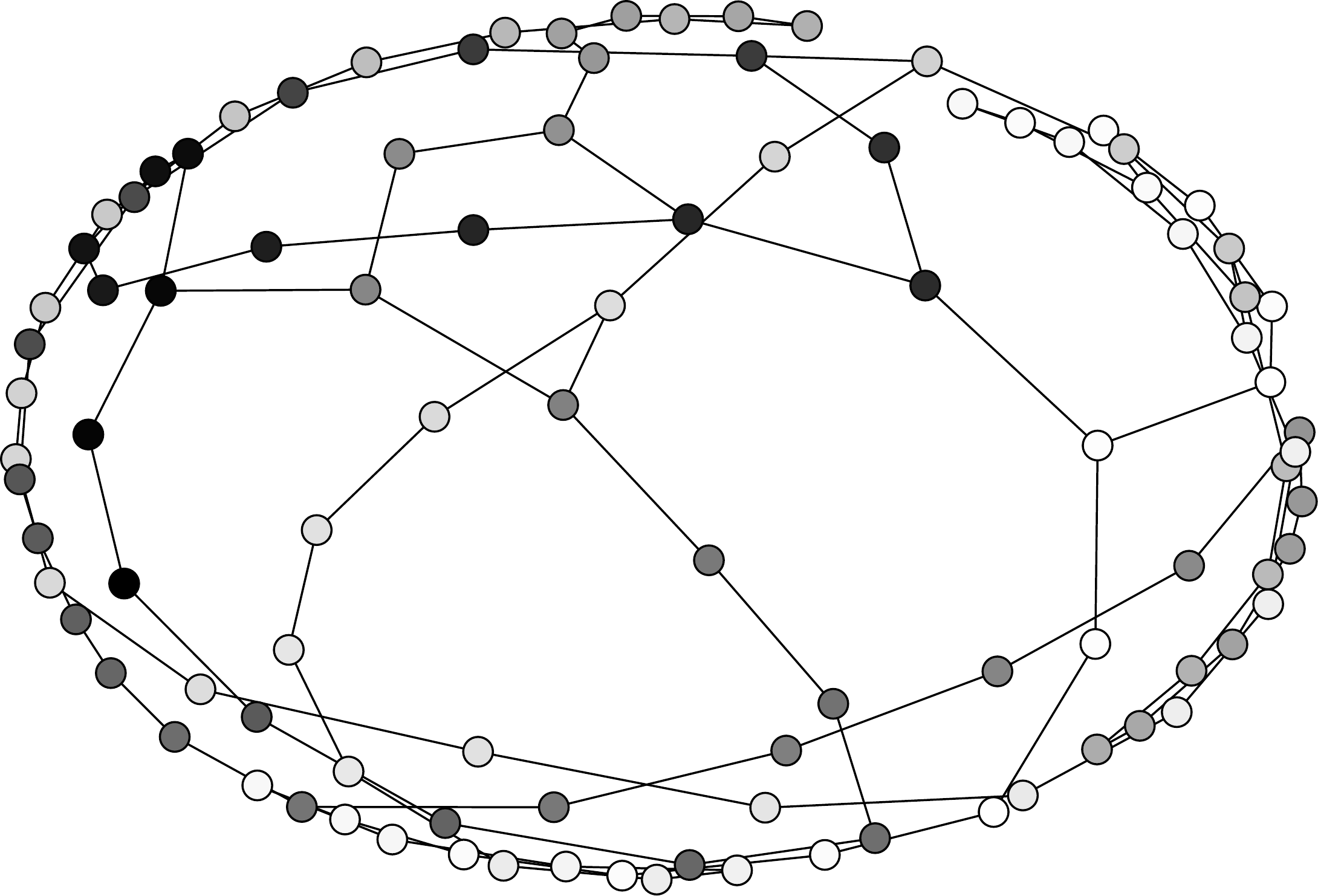}
	}\hspace{0.3cm}
    \subfloat[]{
		\includegraphics[width=0.21\columnwidth]{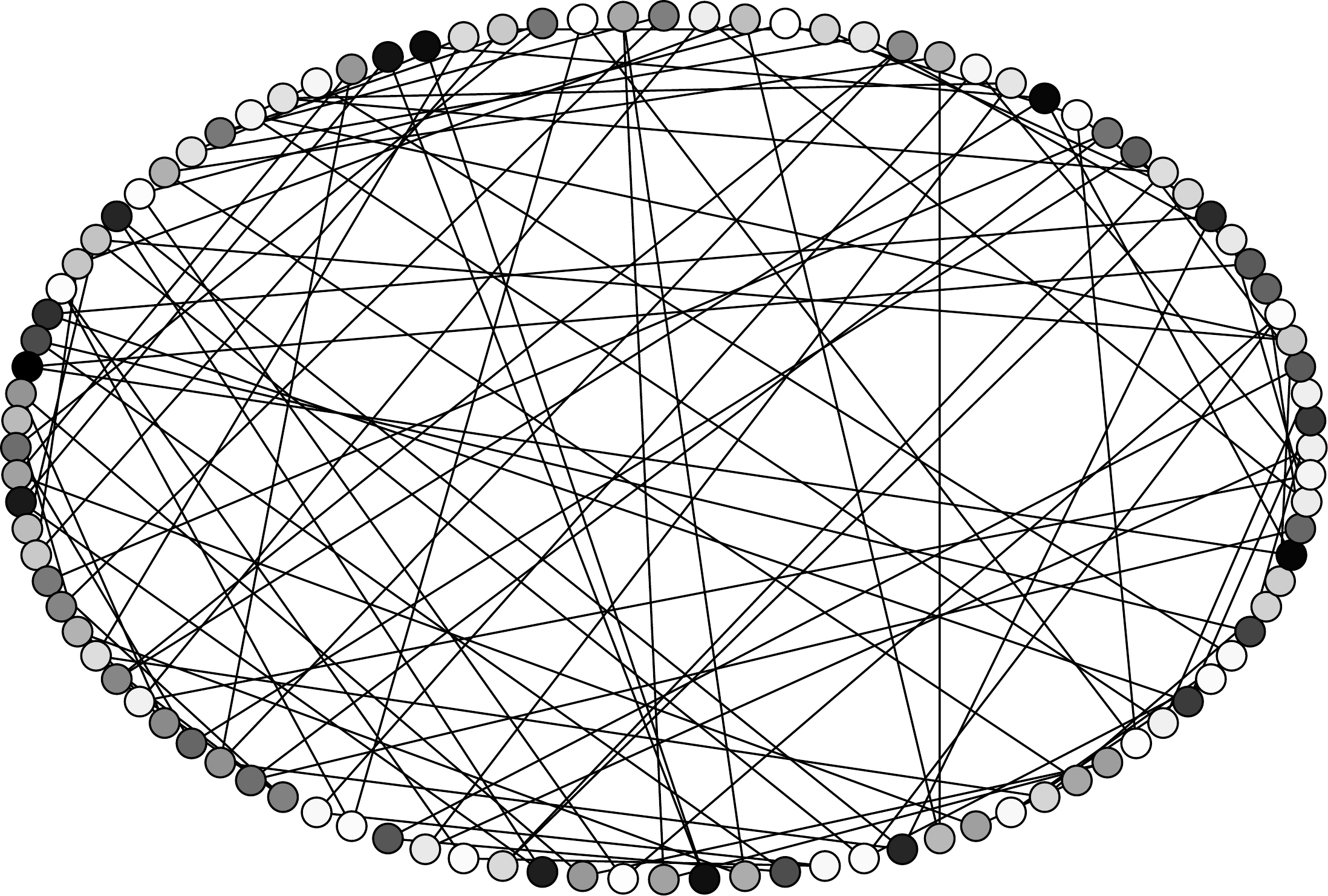}
	}\hspace{0.3cm}
    \subfloat[]{
		\includegraphics[width=0.21\columnwidth]{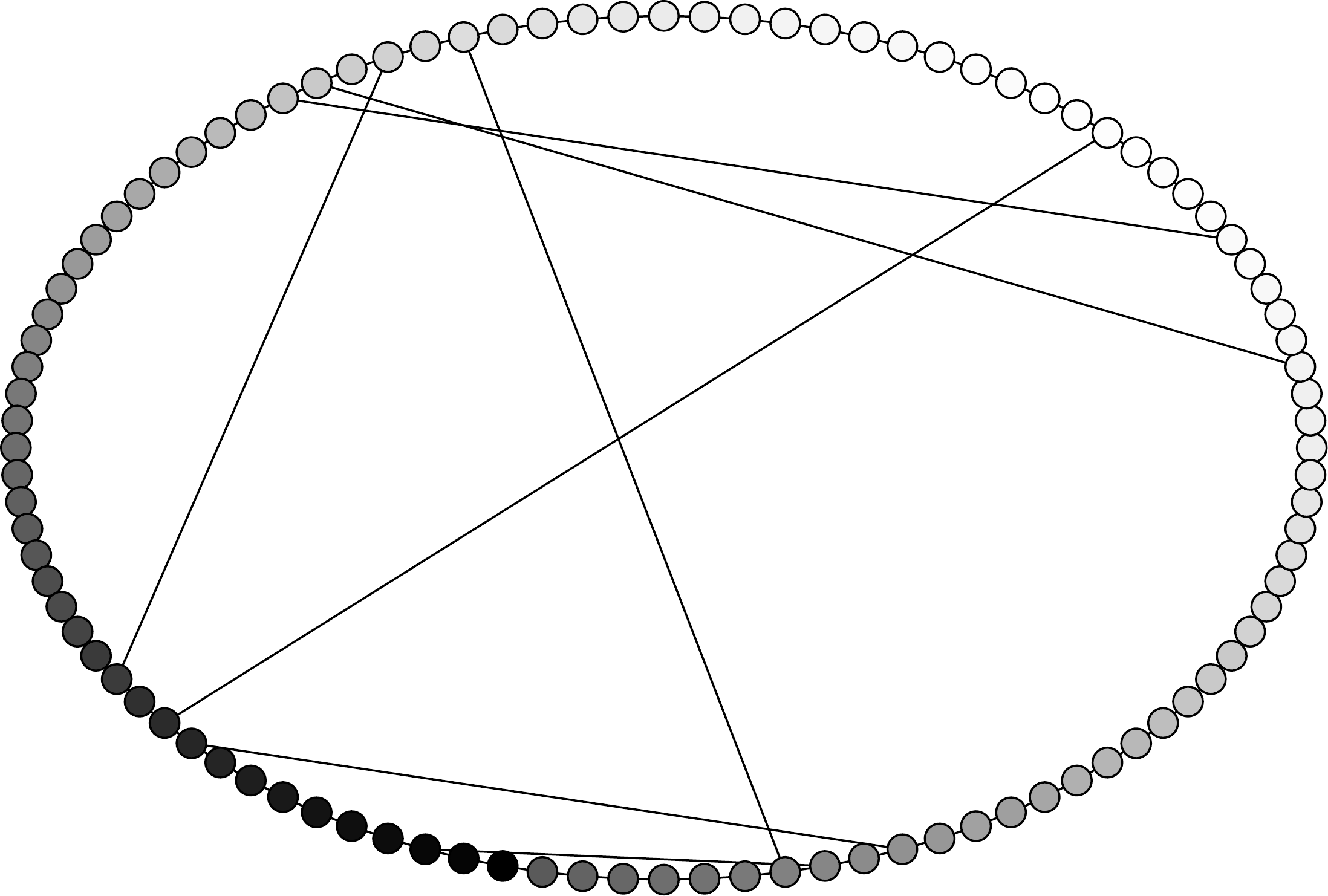}
	}\hspace{0.3cm}
    \subfloat[]{
		\includegraphics[width=0.21\columnwidth]{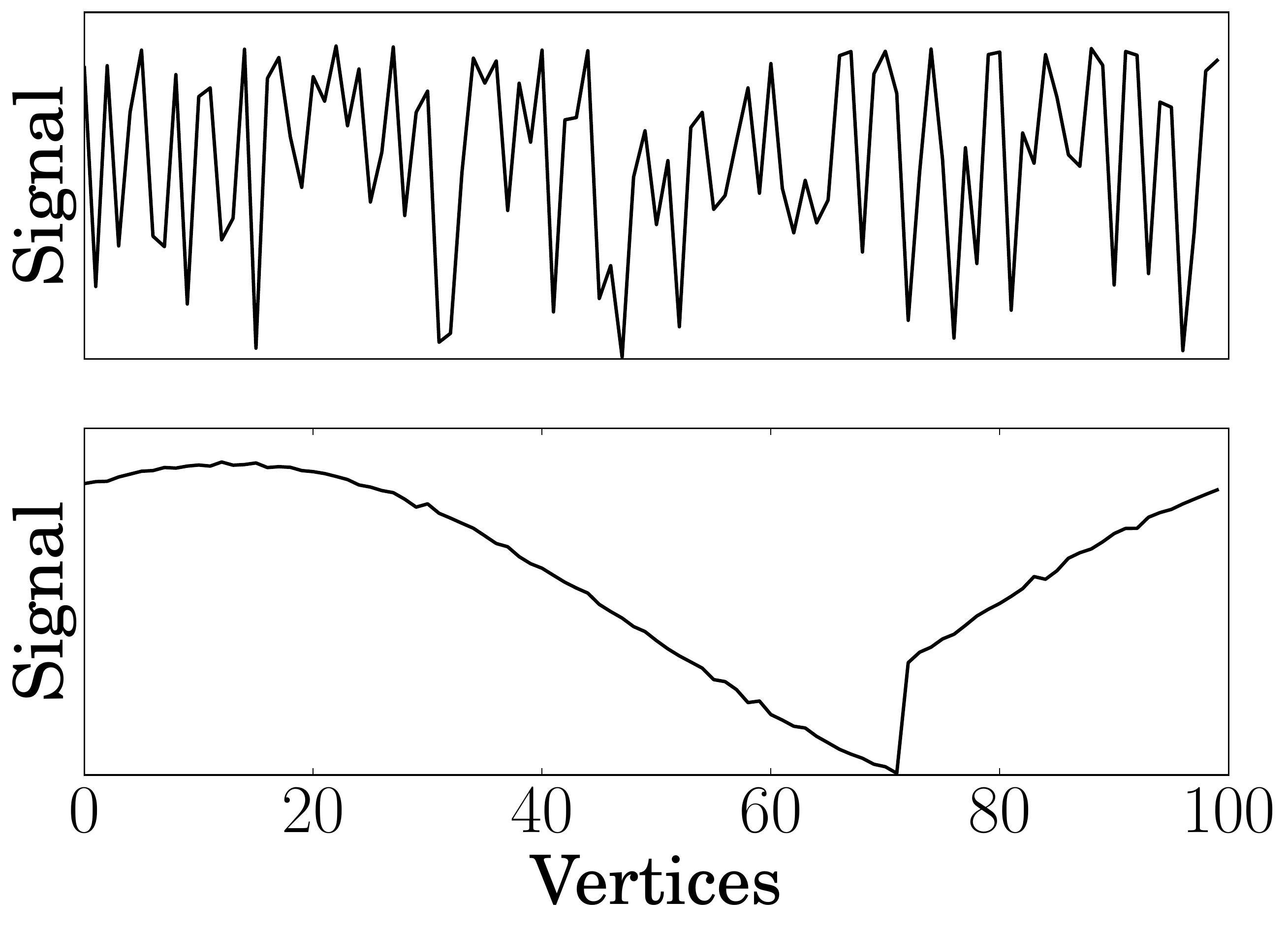}
	}

\caption{\label{fig:son}Example of signals over a network with 50 vertices, with 
a scalar value assigned to each vertex. The color codes the value of the signal 
on each vertex, from black to white. (a) Original network. (b) Circular 
representation of the network according to the index with a random ordering of 
vertices. (c) Circular representation of the network with a suitable ordering of 
vertices. (d) Representation of the signals indexed by the vertices with a 
random vertex ordering (top) and a proper vertex ordering (bottom).}

\end{figure}

A first motivating example comes from the field of signal processing over 
networks, which has been extensively developed in recent years 
\cite{Shuman2013}. Considering a network with an unknown topology, a value is 
assigned to each vertex. This situation could describe for instance a sensor 
network, in which a vertex represents a station which measures a quantity, and 
is in communication with some stations. Different issues come out from this 
example where a proper vertex ordering is of great interest. If we use the 
assumption that the signal is smooth over the network, one question which may 
arise is how to represent the signal according to the vertices in a 
two-dimensional space, to preserve its smoothness. Another point is the 
representation of the network itself, using a linear or circular layout, and how 
to order the vertices to minimize the cross of edges and hence improve the 
visualization. A short example shows that these two questions are linked and can 
be directly addressed if it is possible to obtain an ordering of vertices 
consistent with the topology of the network. Figure~\ref{fig:son} gives an 
example of network with 50 vertices, with a scalar value is assigned to each 
vertex. A random vertex ordering gives both a poor circular representation of 
the network, and a signal with abrupt variations. Conversely, the usefulness of 
a vertex ordering consistent with the structure is clearly visible in the 
circular representation of the network, and gives a smooth representation of the 
signal on the vertices.

\begin{figure}[!ht]
  \centering
  
    \subfloat[\label{subfig:nas_cycle}Cycle graph]{
		\includegraphics[width=0.21\columnwidth]{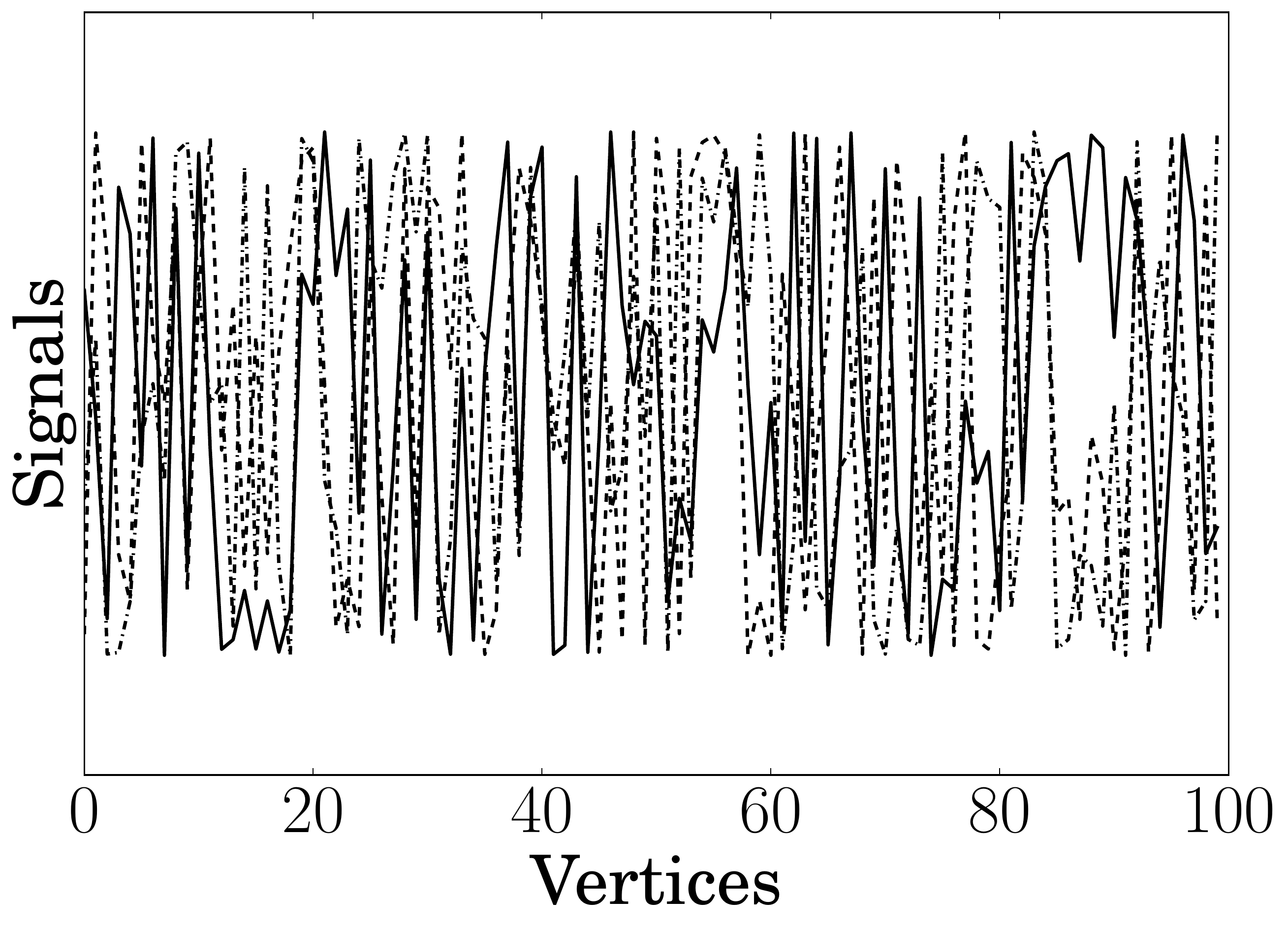}
		\includegraphics[width=0.21\columnwidth]{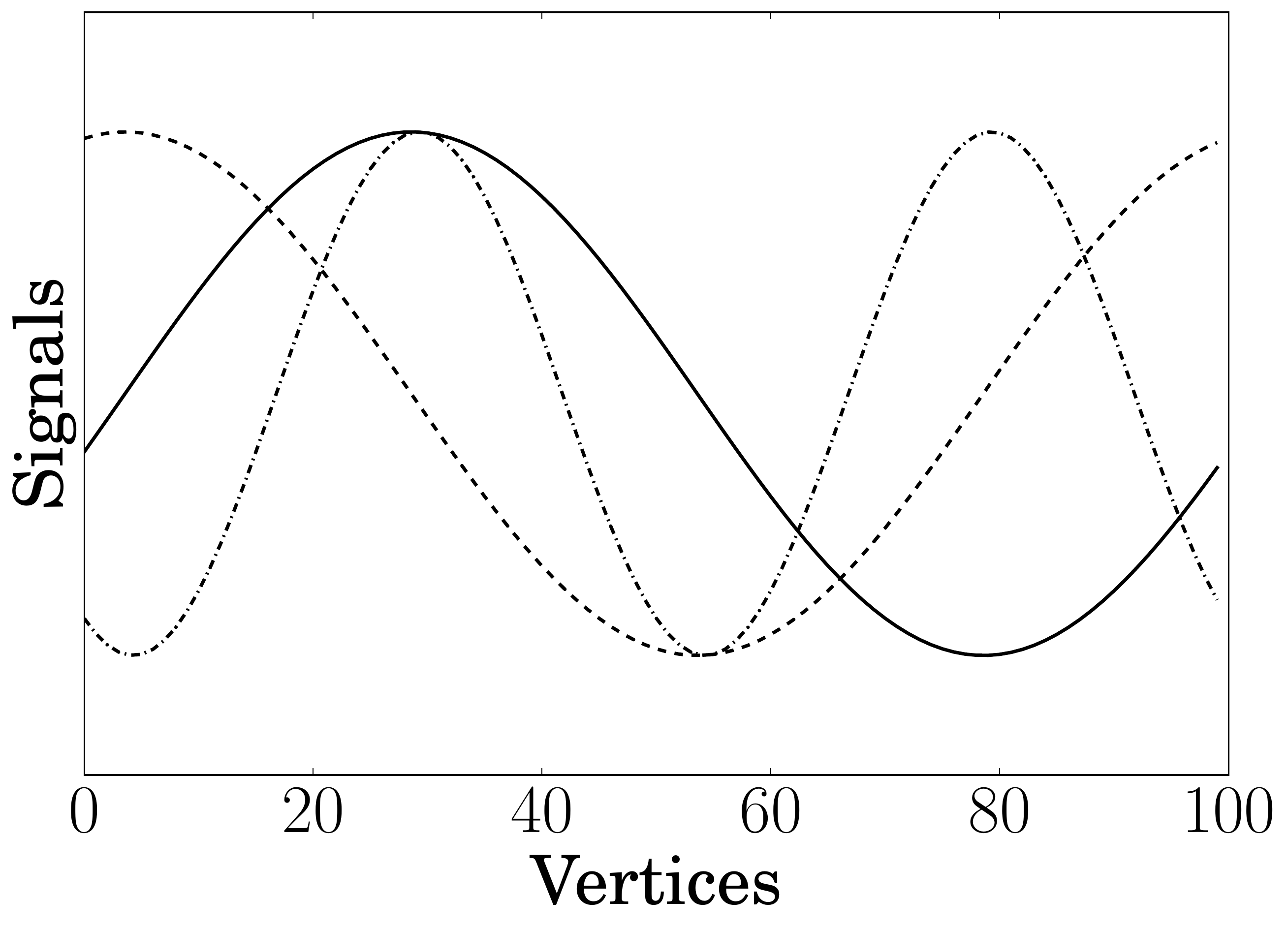}
	} \hspace{1cm}
	\subfloat[\label{subfig:nas_com}Graph with four communities]{
		\includegraphics[width=0.21\columnwidth]{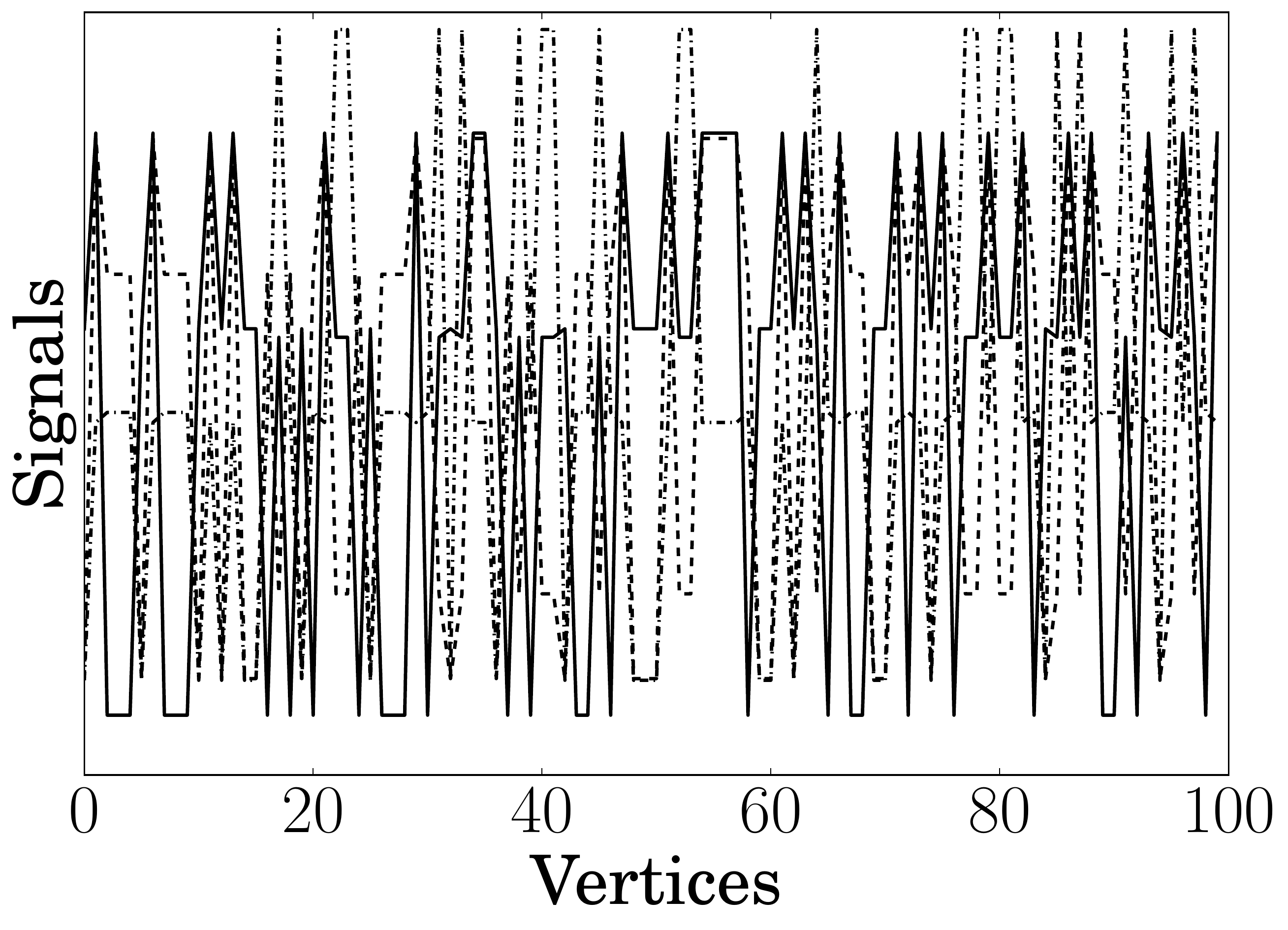}
		\includegraphics[width=0.21\columnwidth]{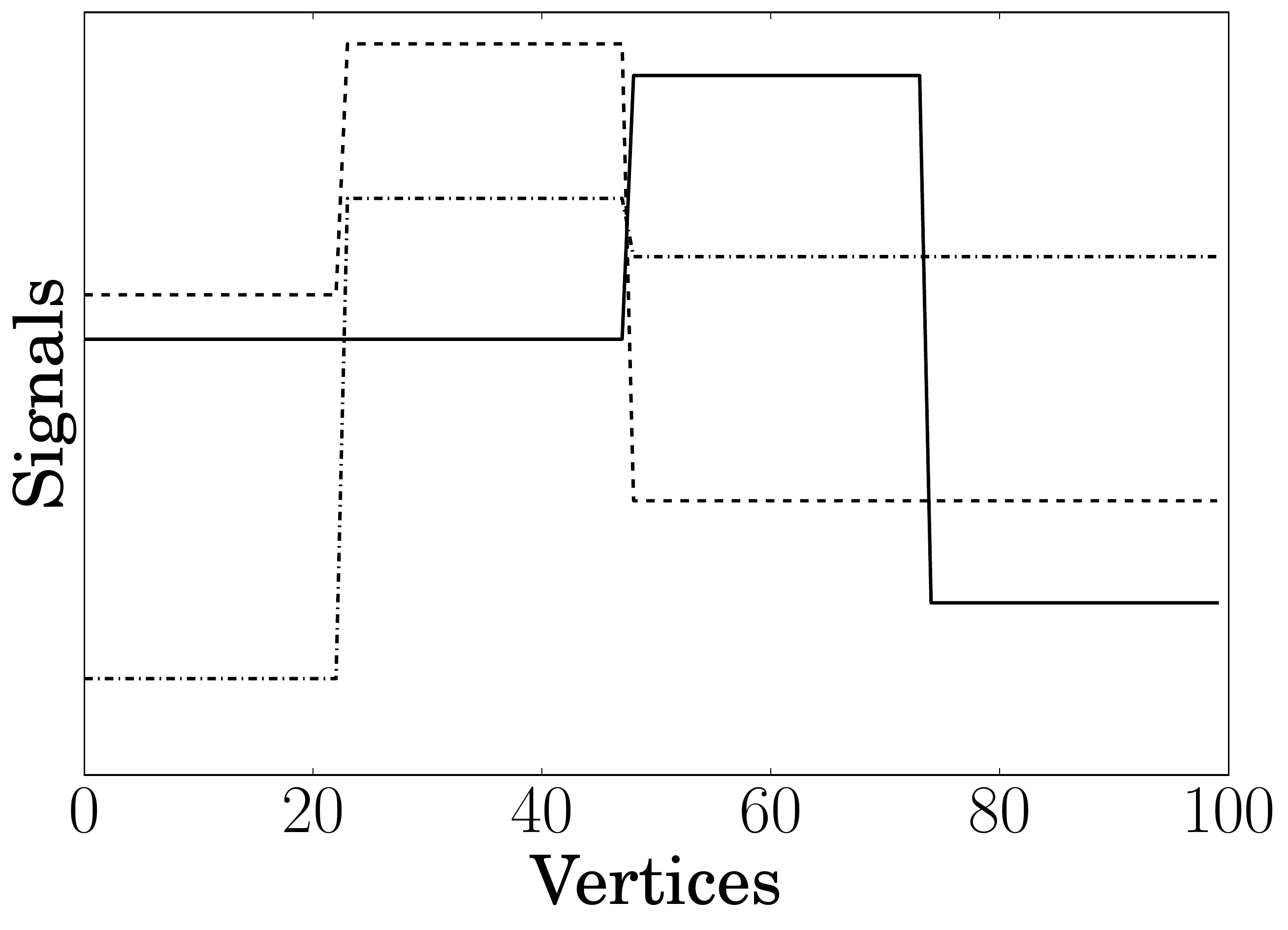}
	}

\caption{\label{fig:nas}Examples of transformation of a network into signals, 
indexed by the vertices using \cite{Shimada2012}. The resulting collection of 
signals is indexed by the vertices. The first three signals are displayed on 
each subplot (left) Random ordering of vertices (right) Suitable ordering of 
vertices.}

\end{figure}

A second example derives from methods of duality between signals and networks, 
which can be taken up using two approaches depending on the object of interest: 
from signals to graphs, for example for the study of time series 
\cite{Campanharo2011}, or from networks to signals, as in \cite{Weng2014} or 
\cite{Shimada2012}. The latter method intends to transform a graph into a 
collection of signals using classical multidimensional scaling \cite{Borg2005} 
and has been extended in \cite{Hamon2013a, Hamon2013b, Hamon2013c, Hamon2014}. 
The aim in these works is to exhibit specific frequency patterns and link them 
with topological properties of the underlying network. A major issue of the 
transformation from graphs to signals concerns the indexation of signals, which 
is based on the vertex order. If two neighboring vertices in this order are not 
adjacent in the graph, then their values in signals might be different and the 
signals are blurred: It leads to abrupt variations of the signals over vertices 
which complicate the spectral analysis and then the monitoring of frequency 
patterns. To smooth the signals, the vertex order must take into account 
adjacency or at least geodesic distance, in other words the vertex order must 
reflect as much as possible the structure of the graph. Fig.~\ref{fig:nas} shows 
two examples of the consequence of a poor indexation to the resulting signals: 
transformation of a cycle graph leads to smooth harmonic oscillations if the 
labeling follows the cycle (\ref{subfig:nas_cycle} right), but to high-frequency 
signals if the labeling is random (\ref{subfig:nas_cycle} left). Likewise, 
transformation of a graph with communities leads to signals with many abrupt 
variations (\ref{subfig:nas_com} left), when an indexation that is consistent 
with the structure in communities highlights plateaus, corresponding to each 
community. It is easy to observe that a spectral analysis on the obtained 
signals might be profitable only if the indexation is consistent with the 
topology. This consistency can be described as follows: Two vertices close in 
the indexation should be close as well in the graph. But as we may have to deal 
with periodic signals, the definition of the proximity (and then of the 
distance) has to be cyclic, i.e. if the vertex at the beginning of the 
indexation and the one at the end have to be close in the network.

We propose in this paper to find such an order that reflects the topology of the 
underlying network. The core of the method consists of the study of a related 
labeling problem, which seek for a mapping from vertices to integers, in such a 
way that an objective function is minimized. This approach is widely made 
explicit in the following.

\subsection{General framework of graph labeling}

Graph labeling consists of assigning labels to vertices or edges of a graph. The 
way in which these labels are classically assigned are driven by the 
minimization of a certain objective function, defined for the purposes of a 
specific application. There exists a wide variety of labeling problems that are 
related to distinctive applications, as described by Diaz et al \cite{Diaz2002}. 
We focus in this article on the labeling of vertices of an unweighted and 
undirected graph, with the objective to find a graph labeling which reflects the 
topology of the graph. As described above, such a labeling shall minimize the 
distances between labels of adjacent vertices, which could be a challenge with 
high stake for many applications. We propose in the following to traverse the 
graph by tackling the cyclic bandwidth sum problem, which consists in minimizing 
the distance between the labels of pairs of connected vertices in the graph.

Let $G=(V,E)$ be a simple connected, unweighted and undirected graph with $V$ 
the set of vertices, and $E$ the set of edges. The number of vertices is noted 
$n = \# V$. Chung \cite{Chung1988} proposed a framework which encompasses many 
graph labeling problems. It is based on a mapping between $V$ and the set of 
vertices of a host graph $H = (N, E_H)$ with $N = \{0, \cdots, n-1\}$. Graph 
labeling problems are then defined as finding the best mapping $\pi$ from $V$ to 
$N$, according to minimization or maximization of an objective function often 
using distances between labels taken between pairs of adjacent vertices of $G$. 
This distance, noted $d_H$, is defined as the length of the shortest path 
between the corresponding vertices in the host graph $H$. Two possible objective 
functions are often considered:

\begin{enumerate}
\item the maximum distance $d_H$ between the labels of two adjacent vertices of 
$G$ is minimized, i.e. it amounts to find a labeling $\hat{\pi}$ such 
that:
  \begin{align}
    \hat{\pi} = \arg\min_\pi\max_{\{u,v\}\in E}d_H(\pi[u],\pi[v])
  \end{align}
\item the sum of distances $d_H$ between all pairs of adjacent vertices of $G$ 
is minimized, i.e. it amounts to find a labeling $\hat{\pi}$ such that: 
  \begin{align}
    \hat{\pi} = \arg\min_{\mathbf{\pi}} \sum_{\{u,v\}\in E} d_H(\pi[u],\pi[v])
  \end{align}
\end{enumerate}

The resulting graph labeling problems have been extensively studied in the case 
where the host graph is a path graph $P$, where $E_P=\lbrace \{i,i+1\}\mid 
i=0\dots n-2\rbrace$. The length of the shortest path between two vertices $u$ 
and $v$ in this graph is given by:
\begin{align}
d_P(\pi[u], \pi[v]) = \vert \pi[u]-\pi[v] \vert
\end{align}
These problems are called bandwidth problem (condition 1) and bandwidth sum 
problem (condition 2).

Lin \cite{Lin1994} and Jianxiu \cite{Jianxiu2001} introduced the problems where 
the host graph is a cycle $C$, where $E_C=\lbrace \{i,i+1\}\mid i=0\dots 
n-2\rbrace \bigcup \{n-1,0\}$. In this case, the distance 
between two vertices $u,v \in V$ is given by:
\begin{align}
\label{eq:dc}
d_C(\pi[u], \pi[v]) = \min\lbrace\vert \pi[u]-\pi[v]\vert, n-\vert \pi[u]-\pi[v] 
\vert\rbrace
\end{align}

The resulting problems are called cyclic bandwidth problem (condition 1) and 
cyclic bandwidth sum problem (condition 2). We focus in this paper on the cyclic 
bandwidth sum problem (\cbsp{}). It is thus defined as the minimization of a 
quantity called cyclic bandwidth sum (\cbs{}):
\begin{eqnarray}
\label{eq:cbsp}
\min_{\mathbf{\pi}} \cbs(G) = \min_{\mathbf{\pi}} \sum_{\{u,v\}\in E} 
d_C(\pi[u], \pi[v])
\end{eqnarray}

Examples of optimal labeling solving Eq.~(\ref{eq:cbsp}) are shown in 
Fig.~\ref{fig:examples} for some standard graphs. We can see that the labeling 
closely follows the structure.

These problems are generally NP-hard problems, as shown by Papadimitriou for the 
bandwidth problem \cite{Papadimitriou1976}  and Lin  for the cyclic bandwidth 
problem \cite{Lin1994}.

\subsection{Related works}

Many works have been done on the study of labeling graph problems: as mentioned 
previously, D\'iaz \cite{Diaz2002} proposed a review of several graph labeling 
problems from an algorithmic point of view. Among these problems, the bandwidth 
problem and bandwidth sum problem have been extensively studied: Papadimitriou 
\cite{Papadimitriou1976} proves the NP-Complenetess of the bandwidth problem, 
highlighting the necessity of heuristics, as the one developed by Cuthill et al 
\cite{Cuthill1969}, to find efficiently a good labeling for these problems. Some 
studies have also been performed on other graph labeling problems, such that 
cyclic bandwidth problem \cite{Lin1994, Romero-Monsivais2013}, antibandwidth 
problem \cite{Calamoneri2006} or cyclic antibandwidth problem \cite{Lozano2013}, 
both in terms of theoretical results and algorithms. Conversely, only few 
results are available in the literature for solving the cyclic bandwidth sum 
problem. Two articles focus on the mathematical aspects of this problem: Jianxiu 
\cite{Jianxiu2001} introduced cyclic bandwidth sum problem and proposed 
theoretical results for some standard graphs, such as wheel or $k$-regular 
graphs, in terms of optimal value of \cbs{} or upper bounds for this value. 
Later on, Chen et al. \cite{Chen2007} extend this work by adding some results, 
for instance for complete bipartite graphs. Whereas these theoretical results do 
not help to get the optimal labeling of a graph, they are nonetheless useful to 
evaluate the quality of a solution of the \cbs{} problem, especially when it is 
obtained thanks to heuristic algorithms.  To the best of our knowledge, only one 
heuristic was proposed to solve the cyclic bandwidth sum problem, published in 
\cite{Satsangi2012} and extended in \cite{Satsangi2013}. The heuristic is based 
on a general variable neighborhood search (\gvns{}). The idea of \gvns{} is to 
change the labeling both globally and locally to descent to local minima of 
\cbs{}, using two distinct phases: A shaking phase in which the labeling is 
changed by applying several operations where the vertices are either shifted, 
reversed, flipped or swapped without taking into account the proximity of 
vertices. This operation enables the algorithm to escape from valleys and to 
browse the solution space. A local search is then performed to descent in a 
valley to a local minima and is performed by switching consecutive vertex or 
swapping adjacent vertices whose edge distance (see Eq.~\ref{eq:dc}) is the 
highest.

\subsection{Outline}

The following sections present the heuristic we developed to address the cyclic 
bandwidth sum problem efficiently, that will be called \algo{} for \algolong{}. 
Section~\ref{sec:proposed_method} sketches the principles of the proposed 
method. Detailed algorithms are presented in Section~\ref{sec:algorithm}, while 
a worst-case complexity study is given in Section~\ref{sec:complexity}. The 
performance of the algorithm is investigated in Section~\ref{sec:results} 
through the comparison of the solution of our algorithm with the theoretical 
results when available, or with \gvns{}. A study is then performed in 
Section~\ref{sec:complex} on graphs exhibiting common properties encountered in 
real-world networks, with a qualitative approach to visually validate the 
performance of the heuristic to discover the structure of complex networks. 
Section~\ref{sec:extension} discusses extensions of the proposed method to 
handle weighted graphs and Section~\ref{sec:conclusion} concludes the paper. 
\makeatletter{}\section{Heuristic to minimize the Cyclic Bandwidth Sum of a graph}
\label{sec:proposed_method}

\begin{figure}[htp]
  \centering
	\subfloat[Path \label{subfig:path}]{
		\includegraphics[width=0.30\columnwidth]{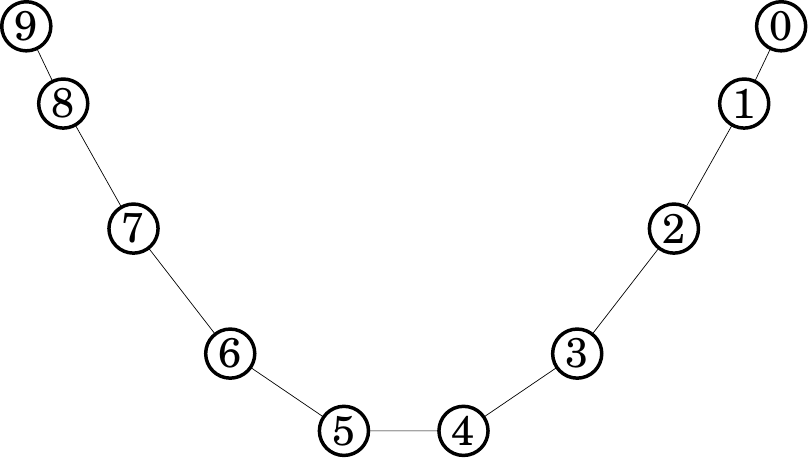}
	}
	\subfloat[Cycle\label{subfig:cycle}] ]{
		 \includegraphics[width=0.30\columnwidth]{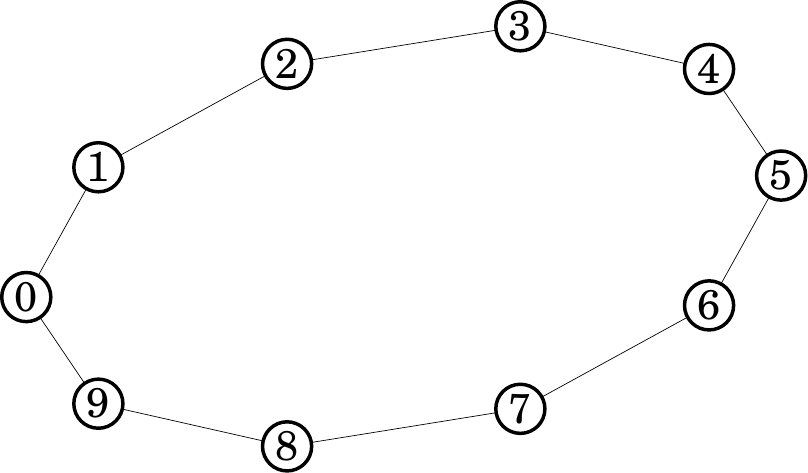}
	}
	\subfloat[Wheel]{
		\includegraphics[width=0.30\columnwidth]{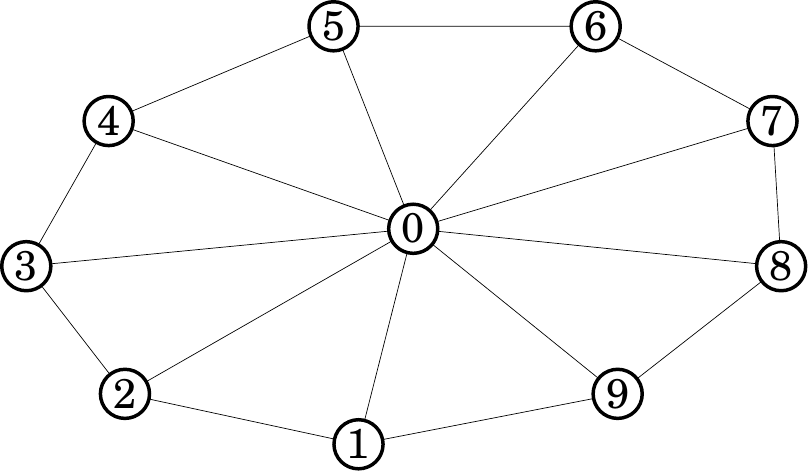}
	}
	
	\subfloat[Complete bipartite (5, 5)]{
		 \includegraphics[width=0.30\columnwidth]{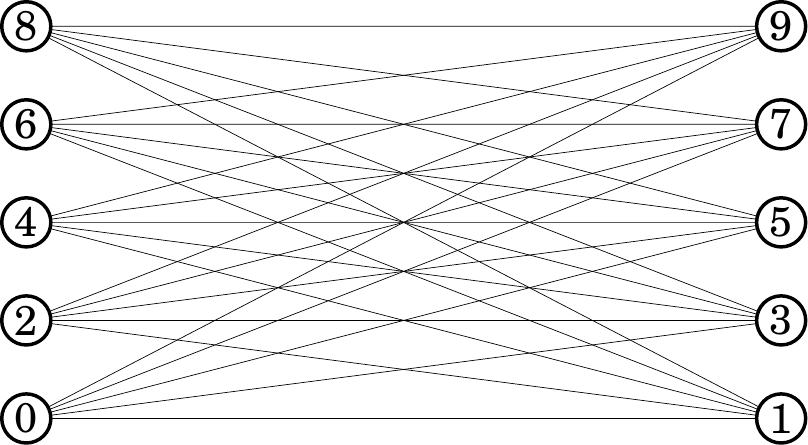}
	}
	\subfloat[Complete bipartite (8, 2)]{
		\includegraphics[width=0.30\columnwidth]{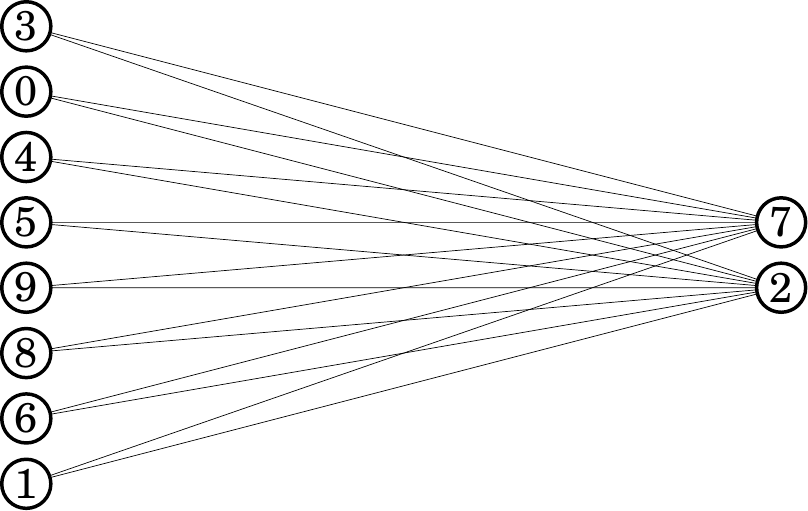}
	}
	\subfloat[Graph with cliques \label{subfig:cliques}]{
		 \includegraphics[width=0.30\columnwidth]{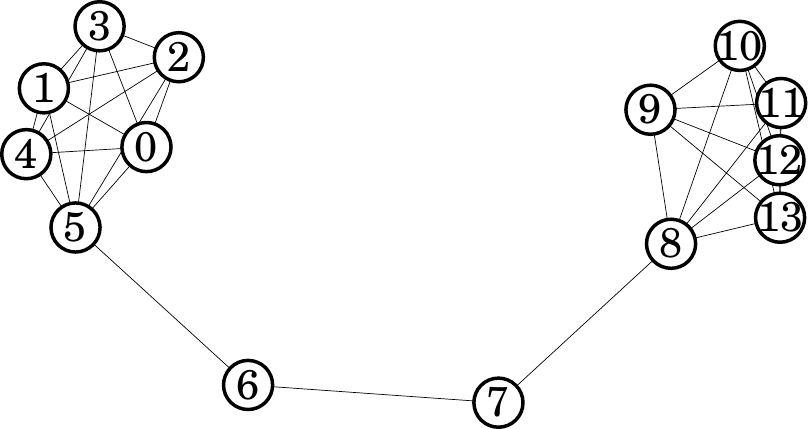}
	}
	
\caption{\label{fig:examples}Examples of standard graphs with optimal     
labeling minimizing the CBS.}
\end{figure}

The aim of the heuristic is to build a labeling by traversing the graph to 
discover its structure. The vertex labels are then constrained by the 
regularities of the structure, which may have multiple forms. For instance, in 
the simple case of a cycle (see Fig.~\ref{subfig:cycle}), the correct behavior 
of the algorithm should be as follows: Starting from a random vertex, it will 
label it and recursively jump to one of its unlabeled adjacent vertices, label 
it with the next integer, and so forth until all vertices are labeled. In the 
less trivial case where the graph is organized by several cliques (see 
Fig.~\ref{subfig:cliques}), the algorithm should browse all the vertices inside 
a clique before jumping to another one. More generally, the algorithm has to 
adapt its search to the structure of the graph, whatever the structure is.

One solution to achieve this goal is to perform a self-avoiding random walk on 
the graph that successively numbers the vertices when they are reached. 
However, this approach has one drawback: the choice of the next vertex 
depends only on the neighborhood of the current vertex, and not on a more 
extended neighborhood. It implies that the random walk has to be controlled to 
avoid going to vertices already numbered, and the walk can stop before visiting 
all the vertices.

The heuristic we propose below fills in the gaps of a random walk and consists 
of a two-step algorithm. The first step performs local searches in order to find 
a collection of independent paths with respect to the local structure of the 
graph, while the second step determines the best way to arrange the paths such 
that the objective function of the \cbsp{} is minimized.

\subsection{Step 1: Guiding the search towards locally similar vertices}
\label{subsec:step1}

The heuristic starts taking as input a graph where each vertex can be referred 
by a unique identifier. The first step consists in finding a collection of paths 
in the graph, that is to say some sequences of vertices consecutively connected. 
The algorithm performs a depth-first search in which the next vertex is chosen 
based on its similarity to the current vertex. This similarity depends on the 
intersection of the two vertex neighborhoods: the more the neighborhoods of the 
two vertices intersect, the closer their labels are. Concretely the search is 
executed as follows: Starting from a vertex, the algorithm jumps to the most 
similar neighbor not yet labeled, and so on until there is no more accessible 
vertices. Then, the algorithm starts a new path from a vertex which has not been 
yet inserted in a path, and then continues to build paths until all the vertices 
are in a path. At the end of this step, a collection of paths is obtained that 
partitions the graph vertex set.

\paragraph{Initialization}

Any vertex not yet inserted in a path can be used as starting node. However, to 
favor the computation of longer paths, vertices that are at the periphery of the 
graph are preferred. The incentive behind this choice lies on the fact that the 
path should start at one of the extremity of the graph if there is one. For 
example, let us consider a simple path graph: Starting from a vertex in the 
middle of the path will generate two paths, although it is obvious that the 
graph can be traversed using a single path. There are several measures to 
determine the centrality of a vertex, that can also be used to find vertices 
that are outer of the graph. We chose the simplest one, to minimize the 
computational cost, by namely using the degree of the vertices: the vertex with 
the smallest degree is selected to start the path.

\paragraph{Construction of a path}

A path is obtained by performing a depth-first search where the next adjacent 
vertex is the one that (1) is not labeled and (2) has a neighborhood that is the 
most similar to the one of the current vertex.  The neighborhood similarity of 
two vertices is evaluated based on the Jaccard index \cite{Jaccard1901}, a 
quantity used to compare the similarity between two sets by looking at the total 
number of common elements (including the considered vertices) over the total 
number of elements. Let $\texttt{Adj}(u)$ returns the adjacent vertices of the 
vertex $u$, i.e. the neighborhood of $u$. The similarity index between the 
vertex $u$ and $v$, noted $J(u,v)$, is defined by:
\begin{align}
\label{eq:jaccard}
 J(u,v) = \frac{\#\Big((\texttt{Adj}(u) \cap  
\texttt{Adj}(v))\cup\{u,v\}\Big)}{\#\Big(\texttt{Adj}(u) 
\cup \texttt{Adj}(v)\Big)} 
\end{align}

This measure is equal to 1 if the two vertices have the same adjacent vertices, 
otherwise it is strictly lower than 1. A value close to 0 means that the total 
number of vertices in the two neighborhoods is much higher than the number of 
common neighbors.

It may happen that two neighbors of the current vertex $u$ have the same 
similarity index with $u$. In this case, the selected vertex is the first vertex 
encountered by the algorithm.

It is preferable that the adjacent vertices of degree 1 that are only adjacent, 
to the current vertex, are not chosen as the following vertices, despite their 
high similarity, because it would end up the path. These vertices are 
immediately inserted after their unique neighbor to guarantee that the vertices 
are as close in the labeling as they are in the graph. However we let the 
traversal 
continue.

\paragraph{End of the search}

The search for a path ends when all the neighbors of the current vertex have 
been inserted in a path. The algorithm starts a new path using the remaining 
vertices, until all the vertices belong to a path.

\subsection{Step 2: Greedy merge of paths}
\label{subsec:step2}

The second step aims at aggregating the paths obtained in Step 1 in a unique 
labeling in such a way that the CBS is minimized. The results of this step is a 
list of vertices, where the position of the vertex in the list gives its label. 
We perform a greedy search that takes the locally optimal choice while merging a 
new path in the partial labeling under construction: The algorithm computes the 
\cbs{} value for the insertion of the path and the reverse of the path at each 
position of the current partial solution and retains the argument that minimizes 
the \cbs{}. The paths are selected in turns according to their length, the 
largest one being selected first. The rational behind this choice is so that to 
broadly explore the space of solutions.

\paragraph{Incremental computing of the CBS}

The evaluating of the \cbs{}, as given in Eq.~(\ref{eq:cbsp}) is demanding 
computationally, as it requires considering every edge of the graph. For each 
insertion of path,  the current \cbs{} is computed twice (ordered and reverse 
ordered) for each possible insertion index of the current labeling. It is thus 
very costly, but can be largely alleviated by observing that, from an index to 
the next one, many edges have the same contribution in the total \cbs{} value. 
From this perspective, we propose an incremental update of the \cbs{} to take 
into account the state before the shift: At each update, only the edges whose 
adjacent vertex labels have been modified are considered.

To explain the incremental computation of the \cbs{}, let us consider the 
insertion of a path, noted $P$, into a labeling, noted $O$, at the index $i$. 
The labeling can be decomposed into three parts: the first part is noted $O_1$ 
and is made of the vertices located before $i$. The vertex right after the index 
of insertion $i$ is noted $k$, while the remaining vertices compose the third 
part called $O_2$. The path $P$ is inserted into the labeling between $O_1$ and 
$k$ when the index of insertion is $i$, and between $k$ and $O_2$ when the index 
of insertion is $i+1$. This is schematically represented in 
Fig.~\ref{fig:incr1}: Line 1 represents the current labeling made of a sequence 
of vertices $O_1$ followed by the vertex $k$ at position $i$ and ended by the 
sequence of vertices $O_2$. $P$ (line 2) is the sequence of vertices that is 
currently inserted at the index $i$ (line 3), i.e. just before vertex $k$. Thus, 
the current labeling begins by the path $O_1$, is followed by $P$, then comes 
the vertex $k$ and the path $O_2$. Line 4 gives the current labeling when $P$ is 
inserted at the position $i+1$, where the vertex $k$ has been shifted from right 
to left. 

\begin{figure}[!h]
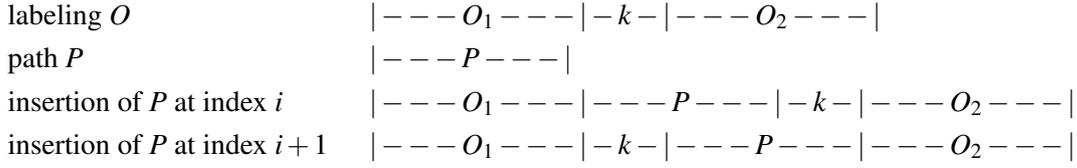

\begin{tabular}{lll}
labeling $O$ & $|---O_1---|-k-|---O_2---|$ \\
path $P$& $|---P---|$ \\
insertion of $P$ at index $i$ & $ |---O_1---|---P---|-k-|---O_2---|$ \\
insertion of $P$ at index $i+1$ & $|---O_1---|-k-|---P---|---O_2---|$
\end{tabular}
\caption{Schema of the insertion of path $P$ in the current 
labeling $O$\label{fig:incr1}}
\end{figure}

As the label of vertices are given by the position of the vertices in the 
labeling, it is clear that from an index to the next one, only the vertices in 
$P$ and $k$ will have different labels. Let $\#P = p$ the number of vertices in 
the path $P$, and $\pi_i[u]$ the label of vertex $u$ when $P$ is inserted at 
index $i$. The changes in the labels for each group of vertices are the 
following:
\begin{align}
\label{eq:pik} \pi_{i+1}[k] &= \pi_{i}[k] - p \\ 
\label{eq:pip} \forall u \in P,\:\pi_{i+1}[u] &= \pi_{i}[u] + 1 \\ 
\label{eq:pio1} \forall u \in O_1,\:\pi_{i+1}[u] &= \pi_{i}[u] \\ 
\label{eq:pio2} \forall u \in O_2,\: \pi_{i+1}[u] &= \pi_{i}[u]
\end{align}
We note $\cbs{}^{(i)}$ the value of the cyclic bandwidth sum when $P$ is 
inserted at index $i$. The computation of $\cbs{}^{(i)}$ can be decomposed 
according to the different groups of vertices defined above:
\begin{align}
\label{eq:cbsi}
\cbs{}^{(i)} &= \cbs{}^{(i)}(O_1, O_1) + \cbs{}^{(i)}(O_2, O_2) + 
\cbs{}^{(i)}(O_1, 
O_2) + \cbs{}^{(i)}(P,P) \\ \nonumber
&+ \cbs{}^{(i)}(k, O_1) + \cbs{}^{(i)}(k, O_2) + \cbs{}^{(i)}(k,P) \\\nonumber
&+ \cbs{}^{(i)}(P, O_1) + \cbs{}^{(i)}(P, O_2)
\end{align}
where $\cbs{}^{(i)}(X,Y)=\sum_{u\in X, v\in Y, \{u,v\} \in E} 
d_C(\pi_i[u],\pi_i[v])$ is the value of \cbs{} when only the edges of the graph 
between the two sets $X$ and $Y$ are considered, with $d_C(\pi[u], \pi[v])$ 
defined in Eq.~\ref{eq:dc}. The definition of the distance $d_C$ shows trivially 
that if the labels of the endpoint vertices of an edge are not shifted or are 
shifted equally, then the value of $d_C$ remains the same:
\begin{align}
 \cbs{}^{(i+1)}(O_1, O_1) &= \cbs{}^{(i)}(O_1, O_1) \\
 \cbs{}^{(i+1)}(O_2, O_2) &= \cbs{}^{(i)}(O_2, O_2) \\
 \cbs{}^{(i+1)}(O_1, O_2) &= \cbs{}^{(i)}(O_1, O_2) \\
 \cbs{}^{(i+1)}(P,P) &= \cbs{}^{(i)}(P,P) 
\end{align}

When the labels of endpoint vertices are not shifted equally, it is necessary to 
consider not only the changes induced by the shift, but also which terms between 
$|\pi[u] - \pi[v]|$ and $n-|\pi[u] - \pi[v]|$ is the minimum, both at index i 
and $i+1$, as it can vary. We prove in the following the results when the 
endpoint vertices are $k$ and a vertex in $O_1$. The other results are given in 
Appendix~\ref{apdx:theorems}.

\begin{theorem}{Edges between $k$ and the vertices of $O_1$}

\label{thm:k01}
Let $u \in O_1$ and $\Delta = \pi_i[k] - \pi_i[u]$. We have:
\begin{enumerate}
\item if $\Delta \leq \frac{n}{2}$ then
  $\cbs{}^{(i+1)}(k,u)=\cbs{}^{(i)}(k,u)-p$.
\item if $\Delta \geq \frac{n}{2} + p$ then
  $\cbs{}^{(i+1)}(k,u)=\cbs{}^{(i)}(k,u)+p$.
\item  if $\frac{n}{2} <  \Delta < \frac{n}{2} + p $  then 
$\cbs{}^{(i+1)}(k,u)=\cbs{}^{(i)}(k,u)+2\Delta - (n+p)$

\end{enumerate}

\end{theorem}

\begin{proof}
For all $u \in O_1$, we have $\pi_{i+1}[u] = \pi_{i}[u] < \pi_{i+1}[k] < 
\pi_{i}[k] $ from Eqs~\ref{eq:pik} and \ref{eq:pio1}. Thus $0 < \pi_{i+1}[k] - 
\pi_{i+1}[u] < \Delta$, allowing for the removal of the absolute value in 
Eq.~\ref{eq:dc}.

Let us consider the case where the minimum term retained to compute 
$\cbs{}_i(u,k)$ in Eq.~(\ref{eq:dc}) is the first term. It means that:

\begin{align}
 \Delta &\leq n - \Delta \Leftrightarrow  \Delta \leq \frac{n}{2}
\end{align}

When $\cbs{}_{i+1}(u,k)$ is considered, the first term is retained in 
Eq.~(\ref{eq:dc}) if:

\begin{align}
 \pi_{i+1}[k] - \pi_{i+1}[u] &\leq n - (\pi_{i+1}[k] - \pi_{i+1}[u]) \\ 
\nonumber
 \Delta - p  &\leq n - (\Delta - p) \\ \nonumber
 2(\Delta - p) &\leq n \\ \nonumber
 \Delta &\leq \frac{n}{2} + p
\end{align}

Symmetrically the second term in the minimum function in Eq.~(\ref{eq:dc}) is 
used for $\cbs{}^{(i)}(u,k)$ if $\Delta \geq \frac{n}{2}$ and for 
$\cbs{}^{(i+1)}(u,k)$ if $\Delta \geq \frac{n}{2} + p$. 

Then, using Eq.~\ref{eq:pik} and Eq.~\ref{eq:pio1}, there are 3 possible cases :

\begin{enumerate}
\item  If $\Delta \leq \frac{n}{2}$, then the first term is retained for 
$\cbs{}^{(i)}(u,k)$ and $\cbs{}^{(i+1)}(u,k)$:
  \begin{align}
\cbs{}^{(i+1)}(k,u) - \cbs{}^{(i)}(k,u) &= (\pi_{i+1}[k] - \pi_{i+1}[u]) -    
(\pi_i[k] - \pi_{i}[u]) \\ \nonumber
&= (\pi_i[k] - p - \pi_i[u]) - (\pi_{i}[k] - \pi_{i}[u]) \\ \nonumber
&= - p
  \end{align}
  
\item If $\Delta \geq \frac{n}{2} + p$, then the second term is retained for 
$\cbs{}^{(i)}(u,k)$ and $\cbs{}^{(i+1)}(u,k)$:

\begin{align}
\cbs{}^{(i+1)}(k,u) - \cbs{}^{(i)}(k,u) &= (n - (\pi_{i+1}[k] - \pi_{i+1}[u]))
- (n - (\pi_i[k] - \pi_{i}[u]) \\ \nonumber
&= -(\pi_i[k] - p - \pi_i[u]) + (\pi_{i}[k] - \pi_{i}[u]) \\ \nonumber
&= p
\end{align}

\item $\frac{n}{2} < \Delta < \frac{n}{2} + p$, then the second term is 
retained for $\cbs{}^{(i)}(u,k)$ and the first term for $\cbs{}^{(i+1)}(u,k)$:

\begin{align}
\cbs{}^{(i+1)}(k,u) - \cbs{}^{(i)}(k,u) &= (\pi_{i+1}[k] - \pi_{i+1}[u]) - (n - 
(\pi_{i}[k] - \pi_{i}[u])) \\ \nonumber
&= (\pi_i[k] - p - \pi_i[u]) - n + (\pi_{i}[k] - \pi_{i}[u]) \\ \nonumber
&= 2\Delta - (n+p)
\end{align}

\end{enumerate}

\end{proof}

\subsection{Comments}

\subsubsection{Influence of the initialization}

The algorithm is completely deterministic and several executions will lead to 
the same solution with a similar input. The algorithm can nevertheless return 
different solutions for a same graph by changing the initial identifiers of the 
vertices:  three steps of the heuristic produce a stochastic behavior and all of 
them originate from the same statement: When a sort is performed, whatever the 
criterion of sorting, if several elements have the same value, then the first 
element encountered by the algorithm is selected before the other ones. This 
happens when (1) the vertices are sorting according to their degree to select 
the first vertex of a path, (2) when several paths have the same length and (3) 
when the path insertion at several positions leads to the same CBS value. The 
stochasticity induced by the initial vertex order is studied in 
Section~\ref{sec:results}, by randomly ordering the vertices $k$ times and 
selecting the minimal value of \cbs{} over the $k$ repetitions, for different 
values of $k$. It shows that when the number of repetitions is high, the 
solution obtained is little bit better. However, the improvement of the 
performance is not really high and indicates the good robustness of the 
heuristic, even with a moderate number of repetitions.

\subsubsection{Local search against global search}

A drawback of the heuristic is that it relies on local searches in the graph. 
Therefore, the algorithm cannot go to a vertex which is not a neighbor of the 
previous one. The labeling is hence really tailored to the structure of the 
graph, as required. Nevertheless, the optimal labeling is sometimes either not 
consistent with the topology of the network for instance if high jumps should 
appear, or it is consistent but using a different organization, not reachable by 
the heuristic. As our main motivation is to follow closely the network 
structure, the obtained labeling can lead to bad results in terms of optimal 
\cbs{}, while finding a path will follow the network. 
\makeatletter{}\section{Detailed algorithm}
\label{sec:algorithm}

The whole algorithm \algo{} comprises the consecutive execution of two steps, 
introduced in Section~\ref{sec:proposed_method}. For readability, the 
algorithm of each step is described separately, respectively in 
Algorithms~\ref{algo:step1} and \ref{algo:step2}. 

From a connected, unweighted and undirected graph $G=(V,E)$ with $n$ vertices,  
the algorithm outputs a one-to-one mapping $\pi$ from $V$ to $\{0, \cdots, 
n-1\}$. We consider in the following a \texttt{List} as a list of elements with 
the associated functions \texttt{List-Insert}($A,a, idx$) which inserts the 
element $a$ in the list $A$ at the index $idx$ (if $idx$ is not given, the 
element $a$ is inserted at the end of the list), \texttt{List-Remove}($A,a$) 
which removes the element $a$ from the list $A$, \texttt{Length}($A$) which 
returns the number of elements of the list $A$, and the function 
\texttt{Reverse}($A$), which returns the list $A$ in the reverse order. The 
function \texttt{Degree}($u$) returns the degree of the vertex $u$ in the graph 
$G$, i.e. the number of vertices adjacent to the vertex $u$. Finally 
\texttt{Adj}($u$) returns the adjacent vertices of $u$.

\begin{algorithm}[!ht]
  \caption{Step 1: Guiding the search towards locally similar vertices}
  \label{algo:step1}
  \begin{algorithmic}[1]
  	\REQUIRE $G = (V, E)$
  	\ENSURE $Paths$
	\STATE $S = \texttt{List}(V)$
	\STATE $Paths \leftarrow \texttt{List}()$
    \WHILE{$S$ is not empty}
    \STATE $u_0 \leftarrow  \arg\min_{u\in S}$\texttt{Degree}($u$)
    \STATE $\texttt{List-Remove}(S, u_0)$
    \STATE $exist\_successors \leftarrow$ \texttt{True}
    \STATE $ P\leftarrow \texttt{List}()$
    \WHILE{$exist\_successors$}
    \STATE \texttt{List-Insert}($P, u_0$)
    \STATE $H\leftarrow \texttt{List}()$
    \FORALL{$v \in \texttt{Adj}(u_0) \cap S$}
    \IF{\texttt{Degree}($v$) = 1}
    \STATE \texttt{List-Insert}($P, v$)
    \STATE $\texttt{List-Remove}(S, v)$
    \ELSE
    \STATE \texttt{List-Insert}($H, v$)
    \ENDIF
    \ENDFOR
    \IF{$H$ is not empty}
    \STATE $u_0 \leftarrow \arg\max_{w\in H}$\texttt{Similarity\_Index}($u,w$)	
    \ELSE
    \STATE $exist\_successors \leftarrow$ \texttt{False}
    \ENDIF
    \ENDWHILE
    \STATE \texttt{List-Insert}($Paths, P$)
    \ENDWHILE
  \end{algorithmic}
\end{algorithm}

Algorithm~\ref{algo:step1} computes the first step of the heuristic as described 
in Section~\ref{subsec:step1}. It consists in building a collection of paths 
containing the vertices of the graph, each  path traversing the graph following 
its structure. Line 1 initializes a list $S$ containing all the vertices of the 
graph, while Line 2 initializes an empty list which will contain the paths . The 
search of paths (Lines 3 to 26) is then performed until all vertices are 
included in a path. A vertex of $S$ minimum degree value is considered (Line 4). 
The selected vertex, noted $u_0$, is removed from $S$ (Line 5) and is the 
starting vertex of the search from Line 8 to Line 24. The path is defined as a 
sequence of vertex added in a list $P$ (Line 7), and is closed when there are no 
more successor available to extend the path or when the depth-first search ends. 
The first step of this loop consists of adding the vertex $u_0$ to the path $P$ 
(Line 9). A new list $H$ is then initialized (Line 10) and will contain the 
potential successors of $u_0$. These successors are selected among the adjacent 
vertices of $u_0$ which are still in the list $S$, i.e. which have not been 
included in a path beforehand (Line 11). For each of the successors, noted $v$, 
if the degree of $v$ is equal to $1$, i.e. the vertex $v$ has only the vertex 
$u_0$ as adjacent vertex, then $v$ is directly added in the path (Lines 12 to 
14). Otherwise, it is added to the list $H$ (Line 16). When all the potential 
successors have been either added to $P$ or $H$, the next vertex to be 
considered is chosen among the elements of $H$, as the one with the highest 
similarity with the current $u_0$ according to Eq~\ref{eq:jaccard} and given by 
the function $\texttt{Similarity\_Index}$. The heuristic then loops using the 
updated value of $u_0$. If $H$ is empty, \texttt{exist\_successors} is set to 
\texttt{False} (Line 22)  and the path is inserted in the list of paths (Line 
25). If $S$ is not empty, then $u_0$ is updated using the procedure described in 
Line 4 and the search of a path from this vertex is repeated. When $S$ is empty, 
the first step is completed.

\begin{algorithm}[!ht]
  \caption{Step 2: Greedy merge of paths}
  \label{algo:step2}
  \begin{algorithmic}[1]
	\REQUIRE $Paths$
    \STATE $Order \leftarrow \arg\max_{P\in Paths}\texttt{Length}(P)$
    \STATE \texttt{List-Remove}($Paths$, $Order$)
    \WHILE{$Paths$ is not empty}
    \STATE $P_0 \leftarrow \arg\max_{P\in Paths} \texttt{Length}(P)$
    \STATE $idx, reverse \leftarrow \texttt{Incremental\_CBS}(Order, P_0)$
    \IF{$reverse$ is \texttt{True}}
	\STATE $\texttt{Insert-List}(Order, \texttt{Reverse}(P_0), idx)$
	\ELSE
	\STATE $\texttt{Insert-List}(Order, P_0, idx)$
    \ENDIF
    \STATE $\texttt{List-Remove}(Paths, P_0)$
    \ENDWHILE
    \STATE
    \STATE $i \leftarrow 0$
    \FOR{$i=0:(n-1)$}
    \STATE $\pi[Order[i]]\leftarrow i$
    \ENDFOR
    \RETURN $\pi$
  \end{algorithmic}
\end{algorithm}

Algorithm~\ref{algo:step2} computes the second step of the heuristic as 
described in Section~\ref{subsec:step2}. A list $Order$ is first initialized as 
the path in $Paths$ with the highest number of elements (Line 1). This path is 
then removed from the list $Paths$, and the algorithm inserts all the remaining 
path in the list $Order$ using a loop from Line 3 to Line 12. The path with the 
highest number of elements is selected (Line 4). The function 
\texttt{Incremental\_CBS} returns the index and the direction of insertion of 
$P$ that minimizes the \cbs{}. Depending on the value of the boolean variable 
$reverse$, the path $P_0$ is inserted reversed (Line 7) or not (Line 9). The 
path is then removed from the list of paths $Paths$ and the heuristic loops 
until all the paths have been inserted in $Order$. As the result, the mapping 
$\pi$ is built using the vertices as keys and the index of the vertices in the 
list $Order$ as values.
 
\makeatletter{}\section{Worst-case complexity of the algorithm}
\label{sec:complexity}

We now examine the worst-case time complexity of the algorithm \algo{} described 
in the previous section, when applied on a graph $G(V,E)$ with $\# V=n$ and $\# 
E=m$. 

We first examine the complexity of Algorithm~\ref{algo:step1}. The set $S$ 
initialized Line 1 can be implemented as a min-priority queue with a binary 
min-heap. The time to build the binary min-heap is $O(n)$. Lines 4 and 5 can be 
done using the \textsc{Extract-Min} function that takes time $O(\log n)$. 
Similarly, the set $Paths$ can be implemented as a max-priority queue with a 
binary max-heap and Line 13 takes in the worst case a time proportional to the 
logarithm of the number of paths, that is in the worst case $O(\log n)$. Using 
aggregate analysis, the \texttt{while} loop in Line 8 is executed at most once 
for each vertex of $V$, since the vertex $u_0$ is removed from $S$ (Line 5). The 
function \texttt{List-Insert} is in constant time. The set $H$ of vertices that 
are adjacent to $u$ and in $S$ is implemented as a max-priority queue using a 
binary heap data structure that makes possible to run 
\textsc{Max\_Heap\_Insert}, that inserts a new element into $H$ (Line 16) while 
maintaining the heap property of $H$ in $O(\log (\#H))$, that is in the worst 
case in $O(\log (\#\texttt{Adj}[u]))$. Thus, the loop on Lines 8-18 in is 
executed $\#\texttt{Adj}[u]$ times and at each iteration (1) the similarity 
index computation takes time $\Theta(\min(\#\texttt{Adj}[u],\#\texttt{Adj}[v]))$ 
and (2) \textsc{Max\_Heap\_Insert} takes time $O(\log(\#\texttt{Adj}[u]))$. 
Therefore, loop Lines 8-18 is in $O(\#\texttt{Adj}[u]^2)$. Line 20 can be done 
using \textsc{Extract\_Max} in time $O(\log(\#\texttt{Adj}[u]))$ and the total 
complexity of Lines 8-24 is in $O(\#\texttt{Adj}[u]^2)$. Consequently, the total 
time is $O(\sum_{u\in V} (\#\texttt{Adj}[u]^2))$. As K. Das established in 
\cite{Das2003} that 
\begin{align}
	\sum_{u\in V} 
\#\texttt{Adj}[u]^2\leq m\left(\frac{2m}{n-1}+n-2 \right)
\end{align}
we can conclude that the total cost of Algorithm~\ref{algo:step1} is in 
$O(n\log n+mn)=O(mn)$

We can also use an aggregate analysis to evaluate the time taken by the 
Algorithm~\ref{algo:step2}. Lines 2 and 5 are in $O(n)$, when almost all the 
vertices have already been merged in \texttt{Order}. \texttt{Incremental\_CBS} 
runs through (1) all the edges between the vertices of the current path $P$ and 
the ones of \texttt{Order} and (2) between the vertex of \texttt{Order} at 
position \texttt{Position} and the other vertices of \texttt{Order}$\cup P$.
\begin{itemize} 
	\item Step (1) takes $O(mn)$ since all the edges of the graph are examined 
when aggregating the analysis over the all paths: the adjacency list of each 
vertex is examined once. Furthermore, for each of these $m$ edges, the $n$ 
positions of \texttt{Order} are evaluated.;
	\item Step (2) is also in $O(mn)$ since aggregating the adjacency lists of 
the vertices of \texttt{Order} leads to the $m$ edges of the graph that are 
evaluated for each of the at most $n$ paths. \end{itemize}
The other 
instructions of the loop are executed in constant time. Therefore, the total 
time spent in Algorithm~\ref{algo:step2} is $O(mn)$. 

Finally, we have that the whole algorithm has a worst case complexity in 
$O(mn)$.
 
\makeatletter{}\section{Computational experiments}
\label{sec:results}

This section describes the computational experiments that we carried out to 
assess the performance of the heuristic \algo{} discussed in the previous 
sections. The aim of this part is to test the ability of the algorithm to obtain 
a good approximate solution for the \cbsp{}, in a reasonable amount of time. 

\subsection{Experimental setup}
\label{subsec:setup}

The assessment of the heuristic is performed on three aspects:

\begin{description}
	\item[\textbf{Performance}] The value of \cbs{} obtained using \algo{} is 
compared with a reference value, chosen as the theoretical results when 
available, or as the value of \cbs{} achieved using an existing method. For 
each 
instance of graph, random identifiers are assigned to the vertices, and the 
heuristic \algo{} is executed, returning the labeling and the value of \cbs{} 
achieved at the end. $30$ repetitions of this process are performed to obtain 
$30$ values of $\cbs{}$ for each instance, in order to check the robustness of 
the results;
	\item[\textbf{Robustness}] The value of \cbs{} obtained using \algo{} is 
compared for different numbers of repetitions of the heuristic, to assess the 
stochasticity of the heuristic and its influence on the results. The process 
described above is repeated $k$ times, and the minimal value over these $k$ 
repetitions is retained. The algorithm is tested for $k \in 
\{10, 20, 50\}$:  For each value of $k$, $30$ repetitions are performed 
as above, to obtain $30$ values of \cbs{};
	\item[\textbf{Execution time}] The average execution time of the algorithm 
is observed to assess the speed of the heuristic. For each instance, the time 
in seconds for the $30$ repetitions of \algo{} is used to obtain an 
average time execution for one repetition.
\end{description}

A comparison is performed between the median value of \cbs{} over the $30$ 
repetitions, noted \emph{median \cbs{}}, with a reference value, noted 
\emph{ref}, which depends on the type of graphs, by computing the relative 
distance \emph{rd}:

\begin{align} 
\text{\emph{rd}} = \frac{\text{\emph{median \cbs{}}} - 
\text{\emph{ref}}}{\text{\emph{ref}}} 
\end{align} 

The sign of \emph{rd} indicates if \emph{median \cbs{}} is greater ($\emph{rd} > 
0$) or lower ($\emph{rd} < 0$) than the reference value, while its value gives 
how far \emph{median \cbs{}} is far from \emph{ref}. For example, $rd=0.80$ 
indicates that \emph{median \cbs{}} is $1.80$ times higher than the reference 
value while $rd=-0.25$ means that the median of the value of \cbs{} is $1.25$ 
times lower than the reference value. Raw data of \emph{median \cbs{}} and 
\emph{ref} are available in Appendix~\ref{sec:detailed_results}.

\algo{} algorithm is implemented in Python 2.7 using the module \emph{Networkx} 
\cite{Hagberg2008}. All the experiments were conducted on a 2.60 GHz 
Intel~Core~i7  with 8 GB of RAM.

\subsection{Datasets}

Eight types of graphs have been considered for the experiments: a detailed 
description of each data set follows.

\subsubsection{Graphs with known \cbs{} optimal value}

\paragraph{\textbf{Path graphs}}
A path graph is defined as a sequence of vertices such that each vertex, except 
the first and the last ones, are linked with its previous and next vertices in 
the sequence. The optimal value for a path $P_n$ of size $n$ is 
$\cbs{}_{\text{opt}}(P_n) = n-1$. The data set consists of all the paths up to 
448 vertices.

\paragraph{\textbf{Cycle graphs}}
A cycle graph is a path whose first and last vertices are linked, 
forming a circular sequence of vertices. The optimal value for a cycle $C_n$ of 
size $n$ is $\cbs{}_{\text{opt}}(C_n) = n$. The data set consists of all the 
cycles up to 448 vertices.

\paragraph{\textbf{Wheel graphs}}
A wheel graph is defined as a cycle whose all vertices are also linked to a 
single 
vertex called \emph{hub}. The optimal value of $\cbs{}$ for a wheel $W_n$ with 
$n$ vertices is  $\cbs{}_{\text{opt}}(W_n) = n + \lfloor \frac{1}{4} n^2 
\rfloor$ as proved in \cite{Jianxiu2001}. The data set consists of all the 
cycles up to 448 vertices.

\paragraph{\textbf{Power graphs of cycles (PGC)}}
The $k$th power of a cycle graph $C_n$ is a graph with $n$ vertices and edges 
such that $u$ and $v$ are linked if and only if the length of the shortest path 
between $u$ and $v$ in $C_n$ is equal or lower than $k$. The optimal value for 
the $k$th power of cycle graph $C_n^k$ is $\cbs{}_{\text{opt}}(C_n^k) = 
\frac{1}{2} nk(k+1)$ as proved in \cite{Jianxiu2001}. The data set consists of 
all the $k$th power of cycles up to 448 vertices, with $k \in \{2, 10\}$.

\paragraph{\textbf{Complete bipartite graphs (CBG)}}
A complete bipartite graph is  composed of two sets with respectively $n_1$ and 
$n_2$ vertices: each vertex of the first set is linked with all the vertices of 
the second set and there is no link between two vertices of 
the same set. Chen et al. \cite{Chen2007} proved that the optimal value of 
\cbs{} for a complete bipartite graph $K_{n_1n_2}$ is given by:
\begin{align}
\cbs{}_{\text{opt}}(K_{n_1n_2}) = 
	  			\left\{
	  			\begin{array}{l l}
	    				\frac{n_1n_2^2 + n_1^2n_2}{4} & \quad \text{if $n_1$ 
and $n_2$ are even}\\
	    				\rule[-1.5ex]{0pt}{5ex}
	    				\frac{n_1n_2^2 + n_1^2n_2 + n_1}{4} & \quad \text{if 
$n_1$ is even and $n_2$ is odd}\\    
	    				\rule[-1.5ex]{0pt}{5ex}
	    				\frac{n_1n_2^2 + n_1^2n_2 + n_1 + n_2}{4} & \quad 
\text{if $n_1$ and $n_2$ are odd}\\
	    				\rule[-1.5ex]{0pt}{5ex}
  	   				\frac{n_1n_2^2 + n_1^2n_2 + n_2}{4} & \quad \text{if $n_1$ 
is odd and $n_2$ is even}\\
	  			\end{array} \right. \nonumber
\end{align}
The data set consists of all the complete bipartite graphs up to 448 vertices 
with three different ratios between $n_1$ and $n_2$: $1$, $3$ and $7$. Only the 
values of $n_1$ and $n_2$ eligible for the desired ratios have been retained.

\subsubsection{Graphs with upper bound of the \cbs{} optimal value}

\paragraph{\textbf{Cartesian products}}
The Cartesian product of two graphs $G = (V_G, E_G)$, with $\#V_G = m$, and $H = 
(V_H, E_H)$, with $\#V_H = n$, noted $G \times H$, is the graph with vertex set 
$V_G \times V_H = \{(u, v) | u \in V_G, v \in V_H\}$. The vertices $(u_G, u_H)$ 
and $(v_G, v_H)$ are adjacent if and only if $u_G = v_G$ and $(u_H, v_H) \in 
E_H$ or $u_H = v_H$ and $(u_G, v_G) \in E_G$.

Jianxiu \cite{Jianxiu2001} proved upper bounds for the optimal value of \cbs{} 
when $G$ and $H$ are either a path, a cycle or a complete graph. A complete 
graph with $n$ vertices, noted $K_n$, is the graph where all pairs of 
vertices are linked. Using the notations given above, we have:
\begin{align}
\cbs{}_{\text{opt}}(P_m \times P_n) &\leq m(n-1) + n^2(m-1), \quad{} m\geq n \\
\cbs{}_{\text{opt}}(C_m \times C_n)&\leq m(n^2+2n-2), \quad{} m\geq n\geq 3 \\
\cbs{}_{\text{opt}}(K_m \times K_n)&\leq \frac{1}{6}mn\left(n^2 + 3n 
\left\lfloor 
\frac{m}{2} \right\rfloor \left\lceil \frac{m}{2} \right\rceil - 1 \right), 
\quad m\geq n \\
\cbs{}_{\text{opt}}(P_m \times C_n)&\leq n(m^2 + m - 1) \\
\cbs{}_{\text{opt}}(P_m \times K_n)&\leq \frac{1}{2}m^2n \left\lfloor 
\frac{n}{2} 
\right \rfloor \left\lceil \frac{n}{2} \right\rceil + n(m-1) \\
\cbs{}_{\text{opt}}(C_m \times K_n)&\leq n \left( 
\frac{1}{2}m^2\left\lfloor\frac{n}{2}\right\rfloor 
\left\lceil\frac{n}{2}\right\rceil + 2m-2 \right)
\end{align}

The data set consists of the Cartesian products of graphs cited above, with $m$ 
and $n$ up to $25$, with the specific constraints on $m$ and $n$ if necessary.

\subsubsection{Graphs with unknown \cbs{} optimal value}

\paragraph{\textbf{Random connected graphs}}
A random graph \cite{Durrett2007} is a graph where the edges between the 
vertices are randomly drawn. The Erd\"os-R\'enyi model has been used to generate 
random graphs: for each pair of vertices, an edge between the two vertices has a 
probability $p$ to appear. The data sets consists of $50$ random graphs built 
as follows: for each value of $p \in \{0.1, 0.3, 0.5, 0.7, 0.9\}$, $10$ 
instances of the Erd\"os-R\'enyi model are generated with a fixed number of 
vertices set to $100$.

\paragraph{\textbf{Harwell-Boeing collection}}
The Harwell-Boeing Sparse Matrix Collection Graphs \cite{Duff1992} consists of 
a set of standard matrices arising from various problems in engineering and 
scientific fields. Graphs are derived from these matrices as follows: Let 
$M_{ij}$ be the element of the $i$th row and the $j$th column of a sparse matrix 
$M$ of size $n \times n$, the resulting graph has $n$ vertices such that there 
is an edge between vertices $i$ and $j$ if and only if $M_{ij} \neq 0$ and 
$i\neq j$. We selected $27$ matrices from this collection, from small graphs (24 
vertices) to large graphs (1454 vertices), representing a wide variety of 
structures. Table~\ref{tab:harwell_perf} gives the list of matrices used with 
the 
number of vertices and the number of edges of the resulting graphs.

\subsection{Performances of the heuristic \algo{}}

\subsubsection{Comparison with known \cbs{} optimal value}

The heuristic \algo{} achieves the optimal value of \cbs{} given by the 
theoretical results for all instances from the following datasets: paths, 
cycles, wheels, power graphs of cycles and complete bipartite graphs.

\subsubsection{Comparison with the upper bound of the \cbs{} optimal value}

The values of \cbs{} obtained using the heuristic \algo{} have been compared 
with the theoretical upper bound given by Jianxiu \cite{Jianxiu2001}. 
Table~\ref{tab:perf_cart_sum} shows a summary of the results: each line concerns 
one type of Cartesian products. The first two columns give the two graphs used 
in the Cartesian products and the third column the number of graph in the 
collection, corresponding for all couple of values $m$ and $n$. The fourth 
column gives the relative distance averaged over all graphs of the collection. 
Finally, the fifth column refers to the table of detailed results in the 
Appendix. 

\begin{table}[!ht]
	\centering
	\small
	\begin{tabular}{|llr|r|l|} 
	\hline
	\multicolumn{1}{|c}{$G$} & 
	\multicolumn{1}{c}{$H$} & 
	\multicolumn{1}{c|}{\emph{\# graphs}} & 
	\multicolumn{1}{c|}{\emph{mean rd}} & 
	\multicolumn{1}{c|}{\emph{Table}} \\ \hline
	
	\makeatletter{}Path & Complete & 400 & -0.84 &\ref{tab:cart_pk_perf} \\ 
Cycle & Complete & 400 & -0.83 &\ref{tab:cart_ck_perf} \\
Complete & Complete & 210 & -0.20 &\ref{tab:cart_kk_perf} \\
Path & Cycle & 400 & 0.50 &\ref{tab:cart_pc_perf} \\ 
Path & Path & 210 & 0.76 &\ref{tab:cart_pp_perf} \\
Cycle & Cycle & 210 & 0.83 &\ref{tab:cart_cc_perf} \\

	\hline
	\end{tabular}
	\caption{\label{tab:perf_cart_sum}Summary of the performance of the 
heuristic \algo{} on the Cartesian products of paths, cycles and complete 
graphs. The first two columns give the two graphs used in the Cartesian 
products 
and the third column the number of graph in the collection. The fourth column 
gives the averaged relative distance. The fifth column refers to the table of 
detailed results in the Appendix.}

\end{table}

The results highlight the different behaviors of the heuristic \algo{} according 
to the topology. When the graph is organized as a succession of linear 
well-structured subgraphs, the heuristic \algo{}  is highly efficient: in the 
case of the Cartesian product of a complete graph and a path or a cycle graph, 
the graph is a succession of cliques in a linear or cyclic arrangement. The 
heuristic \algo{} is then guided by this organization to discover the underlying 
structure, and the value of \cbs{} is consequently very low compared to the 
upper bound. However, when the structure presents regularities but not along a 
linear arrangements, the heuristic \algo{} tends to fail to find the optimal 
value. This happens for the grid (Cartesian products of paths), the torus 
(Cartesian products of cycles) and the cylinder (Cartesian product of a path and 
a cycle): the heuristic \algo{} has many consistent ways to traverse the graph. 
The discovery of the structure is nevertheless performed, but without minimizing 
as much as possible the value of \cbs{}. Detailed results (Tables 
\ref{tab:cart_kk_perf} \ref{tab:cart_pc_perf} \ref{tab:cart_pp_perf} and 
\ref{tab:cart_cc_perf}) show that when the graph is a little unbalanced, for 
example when the number of vertices in $G$ and $H$ is different, the value of 
\cbs{} achieved is better, as the heuristic \algo{} follows this imbalance.

\subsubsection{Comparison with the heuristic \gvns{}}

A comparison with the algorithm developed by Satsangi et al. \cite{Satsangi2012} 
has been performed for the graphs where theoretical results for the optimal 
value of \cbs{} do not exist: the dataset of random graphs and the graphs from 
the Harwell-Boeing collection. To assess the performance of the heuristic 
\algo{}, the heuristic \gvns{} \cite{Satsangi2012} has been executed using the 
code provided by the authors. It has been run using Matlab~R2012b. The algorithm 
has been initially developed to start using the initial identifiers of vertices 
as set at the generation of the graph. These identifiers, especially for the 
graph from the Harwell-Boeing collection, have a topological meaning since these 
matrices describe engineering problems: the value of \cbs{} is then naturally 
low using these identifiers. To check the algorithm in blind conditions, we 
added a randomization of the identifiers of the vertices before applying 
\gvns{}. For each instance, the minimum value of \cbs{} over the $50$ 
repetitions, as defined in \cite{Satsangi2012}, has been retained. Detailed 
results in Appendix show the value of \cbs{} obtained using \gvns{} without and 
with this randomization, respectively noted \emph{\cbs{} w/o r} and \emph{\cbs{} 
w/ r}. We can note that the results without randomization \emph{\cbs{} w/o r} 
are consistent with results presented in the paper \cite{Satsangi2012}. The 
reference value used to compute the relative distance \emph{rd} is the value of 
\cbs{} achieved with randomization (\emph{\cbs{} w/ r}), as we work without any 
prior information on the topology of the graph. Table~\ref{tab:harwell_perf_sum} 
gives a summary of the results for the graphs from the Harwell-Boeing collection 
and Table~\ref{tab:random_perf_sum} for the dataset of random graphs. The two 
first columns give the name of the graph and the number of graphs in the 
collection, while the third column provides the value of \emph{mean rd} averaged 
over all the graphs.

\begin{table}[!ht]
\centering
\small
\subfloat[\label{tab:harwell_perf_sum}Harwell-Boeing collection]{
	
	\begin{tabular}{|lr|r|}
		\hline
		\multicolumn{1}{|c}{\emph{Name}} & 
		\multicolumn{1}{c|}{\emph{\# graphs}} & 
		\multicolumn{1}{c|}{\emph{mean rd}} \\ \hline

		\makeatletter{}bcspwr & 6 & -0.76  \\ 
dwt & 9 & -0.74 \\ 
can & 12 & -0.55  \\ 
 
		\hline

	\end{tabular}
}
\hspace{2cm}
\subfloat[\label{tab:random_perf_sum}Random graphs]{

	\begin{tabular}{|lr|r|}
		\hline
		\multicolumn{1}{|c}{\emph{p}} & 
		\multicolumn{1}{c|}{\emph{\# graphs}} & 
		\multicolumn{1}{c|}{\emph{mean rd}} \\ \hline
	
	\makeatletter{}0.1 & 10 & 0.23 \\ 
0.3 & 10 & 0.10 \\ 
0.5 & 10 & 0.08 \\ 
0.7 & 10 & 0.05 \\ 
0.9 & 10 & 0.02 \\ 
 
	\hline
	\end{tabular}
}
	\caption{Summary of the performance of the heuristic \algo{} on the graphs 
from the Harwell-Boeing collection and on random graphs. The first two columns 
indicate the name of the collection and the number of graphs inside. The fourth 
column gives the average relative distance averaged. Detailed results are given 
in Tables~\ref{tab:harwell_perf} and \ref{tab:random_perf}.}
\end{table}

The results of the heuristic \algo{} on the graphs from the Harwell-Boeing are 
better than those obtained using the heuristic \gvns{}. These graphs are indeed 
highly structured, and the heuristic \algo{} is specially designed to traverse 
this type of graphs. Conversely, the results on random graphs are on average 
less positive compared with those obtained using the heuristic \algo{}, 
especially when the density decreases. By definition, random graphs do not 
exhibit well-structured topology: the heuristic \algo{} has then no support, 
explaining its reduced performance.

\subsection{Robustness of the heuristic \algo{}}
\label{subsec:robustness}

The influence of the stochasticity on the results achieved using the heuristic  
\algo{} is studied, by determining how better the results are when the heuristic 
\algo{} is repeated several times. In this experiment, the median value 
\emph{median 
\cbs} is compared to the minimal value obtained over all the repetitions, set 
as the reference value. Tables~\ref{tab:harwell_sto_sum} and 
\ref{tab:random_sto_sum} show the results of the tests to assess the robustness 
of the heuristic \algo{}. For each table, the first two columns give the name 
and the number of graphs in the collection, while the next three columns give 
the 
relative distance averaged over all graphs for the different values of $k$.

\begin{table}[!ht]
	\centering
	\small
	\subfloat[\label{tab:harwell_sto_sum}Harwell-Boeing collection]{
		\begin{tabular}{|lr|rrr|}
			\hline
			\multicolumn{2}{|c|}{\textbf{Collection}} & 
			\multicolumn{3}{c|}{\textbf{Relative distance}} \\
			\hline
			\multicolumn{1}{|c}{\emph{Name}} & 
			\multicolumn{1}{c|}{\emph{\# graphs}} & 
			\multicolumn{1}{c}{$k=10$} & 
			\multicolumn{1}{c}{$k=20$} & 
			\multicolumn{1}{c|}{$k=50$} \\ \hline

			\makeatletter{}bcspwr & 5 & 0.09 & 0.07 & 0.06 \\ dwt & 6 & 0.10 & 0.08 & 0.05 \\ can & 9 & 0.10 & 0.08 & 0.06 \\  
			\hline

	\end{tabular}
	}
	\subfloat[\label{tab:random_sto_sum}Random graphs]{	
	\begin{tabular}{|lr|rrr|}
		\hline
		\multicolumn{2}{|c|}{\textbf{Collection}} & 
		\multicolumn{3}{c|}{\textbf{Relative distance}} \\
		\hline
		\multicolumn{1}{|c}{$p$} & 
		\multicolumn{1}{c|}{\emph{\# graphs}} & 
		\multicolumn{1}{c}{$k=10$} & 
		\multicolumn{1}{c}{$k=20$} & 
		\multicolumn{1}{c|}{$k=50$} \\ \hline

		\makeatletter{}0.1 & 10 & 0.00 & 0.00 & 0.00 \\ 
0.3 & 10 & 0.00 & 0.00 & 0.00 \\ 
0.5 & 10 & 0.00 & 0.00 & 0.00 \\ 
0.7 & 10 & 0.00 & 0.00 & 0.00 \\ 
0.9 & 10 & 0.00 & 0.00 & 0.00 \\ 
 
		\hline
	\end{tabular}
	}
	\caption{Summary of the study of the robustness of the heuristic \algo{} on 
the graphs from the Harwell-Boeing collection and on random graphs. The first 
two columns indicate the name of the collection and the number of graphs in the 
collection. The next three columns report the relative distance between the 
median value of the minimal value of \cbs{} over $k$ repetitions and the minimal 
value obtained over all repetitions, for three values of $k$. Detailed 
results are given in Tables~\ref{tab:harwell_sto} and \ref{tab:random_sto}.}
\end{table}

The results show that in both data sets, the heuristic \algo{} is robust since 
the increase of the number of repetitions does not highly improve the obtained 
results: for the graph from the Harwell-Boeing collection, the advantage is 
slight, while there is no improvement in the case of random graphs. These two 
examples show that even if the heuristic \algo{} has a stochastic behavior, it 
has only a slight influence on the results.

\subsection{Execution time}

The speed of the heuristic \algo{} is assessed by looking at the average 
execution time for different instances of graphs. Table~\ref{tab:time_standard} 
shows the average execution time in seconds of the heuristic \algo{} for 
different values of $n$ on different types of graphs.
\begin{table}[!ht]
	\centering
	\small
	\begin{tabular}{|r|rrrrrrrr|}
		\hline
		\multicolumn{1}{|c|}{$|V|$} &
		\multicolumn{1}{c}{Path} & 
		\multicolumn{1}{c}{Cycle} & 
		\multicolumn{1}{c}{Wheel} & 
		\multicolumn{1}{c}{PGC 2} & 
		\multicolumn{1}{c}{PGC 10} & 
		\multicolumn{1}{c}{CBG 1} & 
		\multicolumn{1}{c}{CBG 3} & 
		\multicolumn{1}{c|}{CBG 7} \\ \hline

		\makeatletter{}8 & $<0.01$& $<0.01$& $<0.01$& $<0.01$& $<0.01$& $<0.01$& $<0.01$& $<0.01$\\ 
64 & $<0.01$& $<0.01$& 0.01 & 0.02 & 0.05 & 0.07 & 0.46 & 0.17 \\ 
128 & 0.02 & 0.02 & 0.02 & 0.03 & 0.1 & 0.26 & 2.02 & 1.03 \\ 
192 & 0.01 & 0.01 & 0.03 & 0.04 & 0.14 & 0.75 & 10.2 & 3.35 \\ 
256 & 0.03 & 0.04 & 0.04 & 0.03 & 0.18 & 1.58 & 24.7 & 7.86 \\ 
320 & 0.01 & 0.03 & 0.08 & 0.05 & 0.24 & 2.82 & 32.7 & 15.3 \\ 
384 & 0.03 & 0.04 & 0.12 & 0.07 & 0.3 & 4.49 & 41.9 & 26.2 \\ 
448 & 0.05 & 0.05 & 0.13 & 0.05 & 0.32 & 6.61 & 52.5 & 42.9 \\  

		\hline
	\end{tabular}
	\caption{\label{tab:time_standard}Averaged execution time in seconds of 
the heuristic \algo{} on several types of graphs for different values of $n$.}

\end{table}

When the topology of the graph is simple, as for instance for the paths, the 
cycles or the wheels, the algorithm goes quite fast, running in less than one 
second. The computational cost of the algorithm hugely increases when the graph 
is the complete bipartite graph, which can be explained by the peculiar 
structure of these graphs: Indeed, the algorithm will first of all compute a 
first path containing all the vertices of the smaller subset and the same number 
of vertices in the other one. All the remaining vertices will be considered as a 
path of length $1$ (since they are isolated when the smaller subset is removed), 
and the algorithm will spend a huge amount of time to merge one by one all these 
vertices with the first path. This very greedy step makes explode the 
computational cost when $n$ increases for these graphs. 

Detailed results in Appendix \ref{sec:detailed_results} display the average 
execution time for one execution of the heuristic \algo{} for the other 
collections. These results show that the heuristic \algo{} is well-adapted for 
real-world graphs from hundreds to thousands of vertices. The computational cost 
prevents however the use of the heuristic \algo{} when the graph has millions of 
vertices. It is nevertheless, in view of our computational results, the fastest 
solution to minimize the cyclic bandwidth sum on graphs. 
\makeatletter{}\section{Applications on complex networks}
\label{sec:complex}

We focused in this section on the assessment of the performances of the 
heuristic \algo{} on complex networks. Contrary to the previous section, the 
experiments are guided by the motivation to discover the structure of the 
network.  We then propose to use the value of \cbs{} only to assess the 
robustness of the heuristic, as performed previously. Then a visual validation 
of the performance of the heuristic is done to check the match between topology 
and labeling.

\subsection{Properties of complex networks}

Three well-known properties of real-world complex networks are often 
encountered: presence of communities, scale-free property and 
small-worldness.  

\paragraph{Graph with communities (COM)}
One definition of a community is a group of vertices such that the number of 
edges between the vertices of a community is significantly higher than between 
vertices belonging to different communities. This property is for instance 
well-known in social networks, where people tends to belong to groups of people. 
We considered the following stochastic block model to build a graph with three 
communities: each vertex is randomly assigned to one of the three communities. 
For each pair of vertices, an edge exists with a probability $p_{\text{intra}}$ 
if the two vertices belong to the same community, and with a probability 
$p_{\text{inter}}$ otherwise, with $p_{\text{inter}} < p_{\text{intra}}$. In our 
experiments, we set $p_{\text{inter}} = 0.01$ and $p_{\text{intra}} = 0.9$.

\paragraph{Scale-free network (SF)} 
Scale-free property means that the distribution of degrees follows a power law. 
This property has been highlighted in many real-world networks, such that social 
networks. It exists several methods to generate a scale-free network, among them 
we used the Barab\'asi-Albert model \cite{Albert2002}: from an initial connected 
graph, vertices are sequentially added, and attached to one existing vertex, 
chosen according to a probability which depends on the degree of the vertex. The 
higher the degree is, the higher the probability. The new vertices tend then to 
connect to vertices with high degree.

\paragraph{Small-world networks (SW)} 
The small-worldness is a property of networks whose the average shortest path 
length is small with regard of the number of vertices. The generative model used 
is the Watts-Strogatz model \cite{Watts1998}: starting from a regular ring 
lattice of degree $k$, an edge between two unlinked vertices is drawn with a 
probability $p$. We choose $k=4$ and $p=0.1$.

\subsection{Robustness of the heuristic}

As previously performed in Section~\ref{subsec:robustness}, the robustness of 
the heuristic is tested on complex networks with $100$ vertices. 
Tab.~\ref{tab:complex_sto_sum} shows the results of the tests. For each table, 
the first two columns give the name and the number of graphs in the collection, 
while the next three columns give the relative distance averaged over all 
graphs 
for the different values of $k$.

\begin{table}[!ht]
	\centering
	\small
	\begin{tabular}{|lr|rrr|}
		\hline
		\multicolumn{2}{|c|}{\textbf{Collection}} & 
		\multicolumn{3}{c|}{\textbf{Relative distance}} \\
		\hline
		\multicolumn{1}{|c}{\emph{Name}} & 
		\multicolumn{1}{c|}{\emph{\# graphs}} & 
		\multicolumn{1}{c}{$k=10$} & 
		\multicolumn{1}{c}{$k=20$} & 
		\multicolumn{1}{c|}{$k=50$} \\ \hline

		\makeatletter{}SF & 10 & 0.02 & 0.02 & 0.01 \\ 
SW & 10 & 0.00 & 0.00 & 0.00 \\ 
COM & 10 & 0.00 & 0.00 & 0.00 \\ 
 
		\hline

\end{tabular}

\caption{\label{tab:complex_sto_sum}Summary of the study of the robustness of 
the heuristic \algo{} on complex networks.  The first two columns indicate the 
name of the collection and the number of graphs in the collection. The next 
three columns report the relative distance between the median value of the 
minimal value of \cbs{} over $k$ repetitions and the minimal value obtained over 
all repetitions, for three values of $k$. Detailed results are given in 
Tab.~\ref{tab:complex_sto}.}
\end{table}

The results confirm the good robustness of the heuristic: for complex networks 
with three different properties, the value of \cbs{} slightly improved for the 
scale-free network when the number of repetitions increases, and is at a 
standstill for the other networks.

\subsection{Visualization of networks before and after labeling}
A visual validation of the consistence between labeling and topology is 
succinctly performed in this section. For one instance of each type of networks, 
the graph is displayed using a layout consistent with the topology. The label of 
each vertex is displayed using a shade of gray defined such that two close 
colors denotes a short distance within the meaning of Eq.\ref{eq:dc}. 
Figs.~\ref{fig:sf}, \ref{fig:sw} and \ref{fig:com} show for the three type of 
networks two representations of the graphs, before and after relabeling. 

\begin{figure}[htp]
	\centering
 
	\includegraphics[width=0.3\columnwidth]{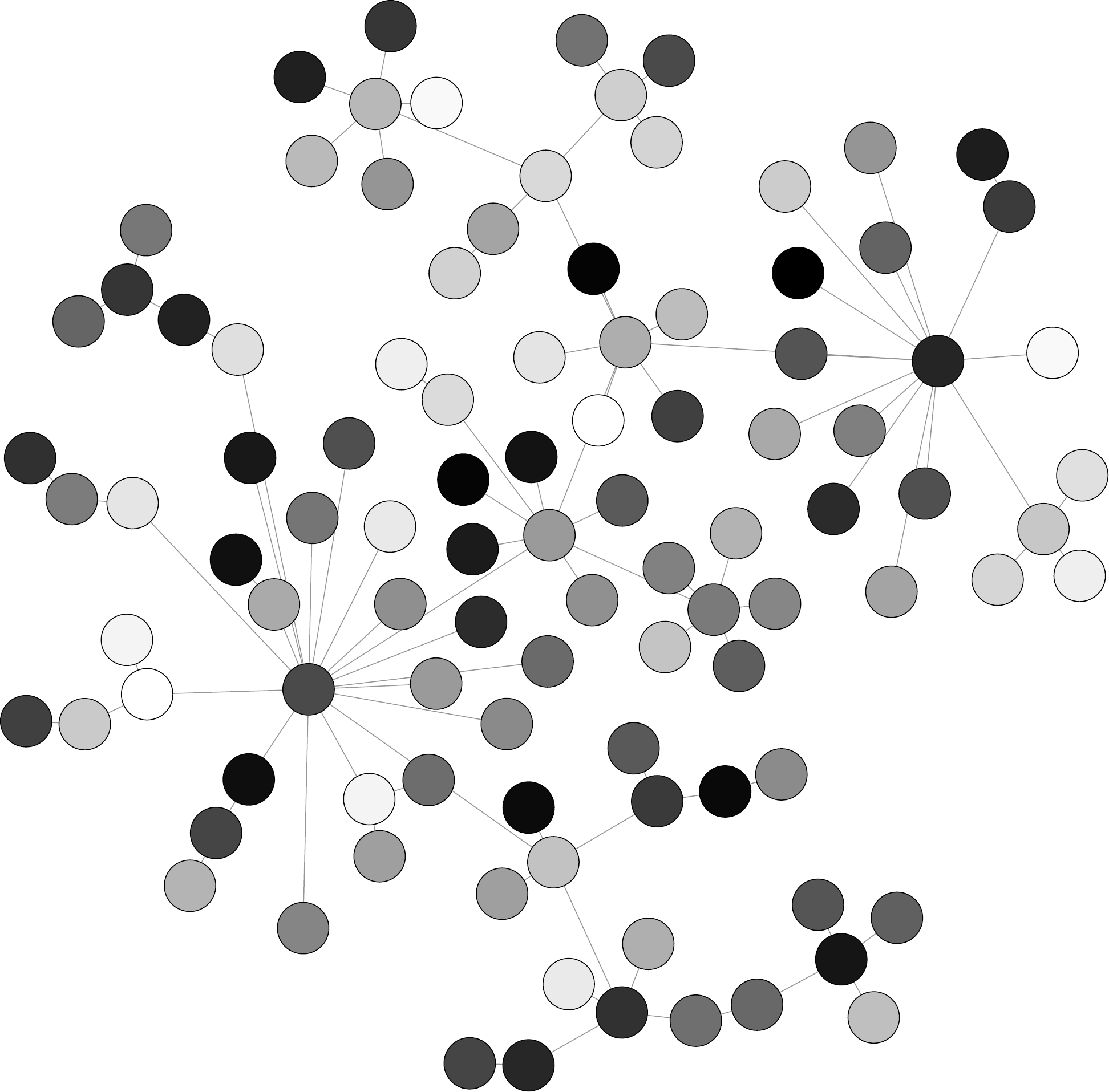} \hspace{1cm}
	\includegraphics[width=0.3\columnwidth]{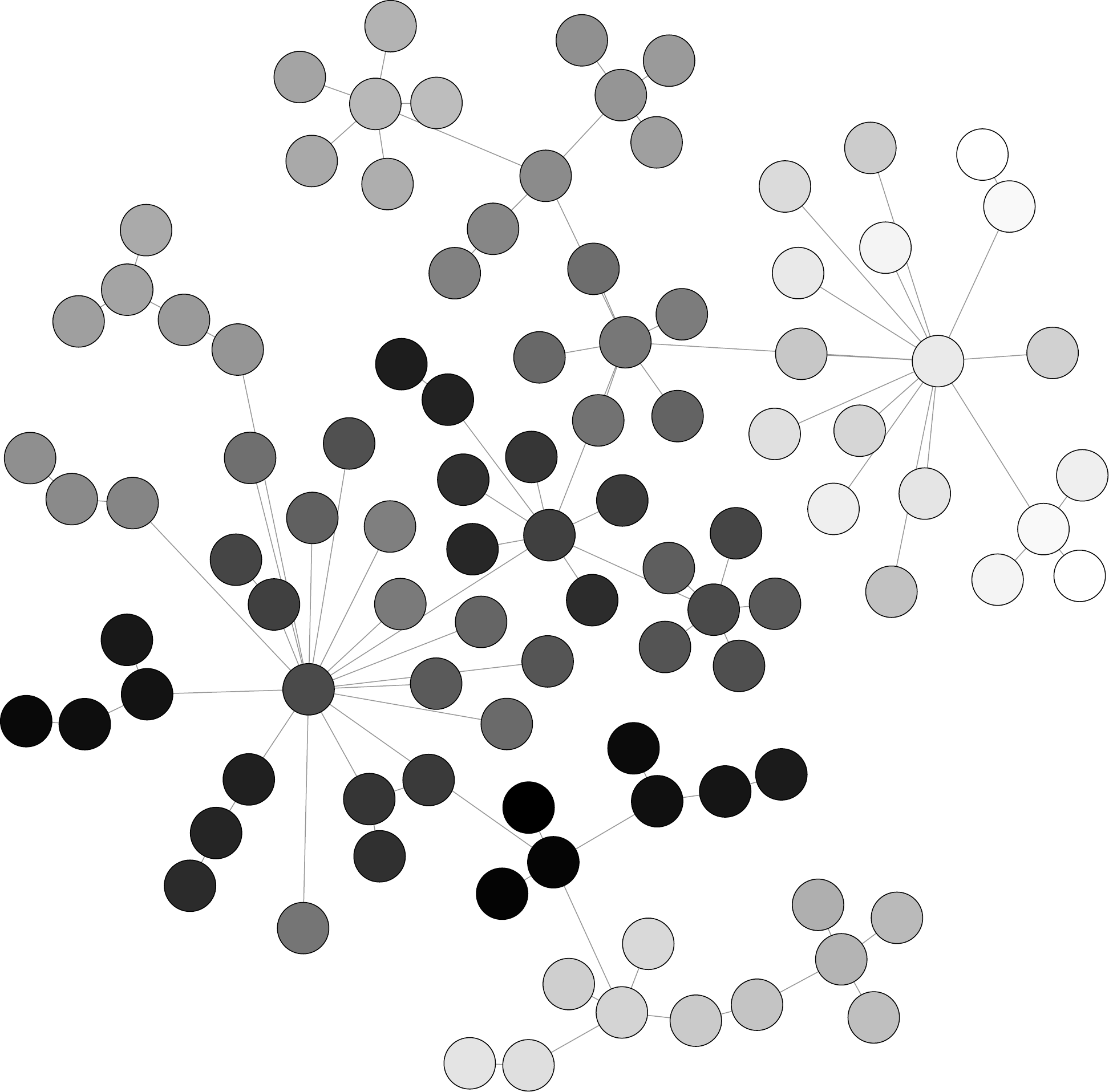}

	\caption{\label{fig:sf}The left and right figures show the vertices of a 
scale-free network respectively before and after applying the heuristic \algo{} 
to relabel vertices. The color of the vertex depends on the label: two close 
colors means that the labels are close as well.\\ \cbs{} value before 
relabeling: 2734 ; \cbs{} value after relabeling: 375.}
\end{figure}

\begin{figure}[htp]
	\centering
	\includegraphics[width=0.3\columnwidth]{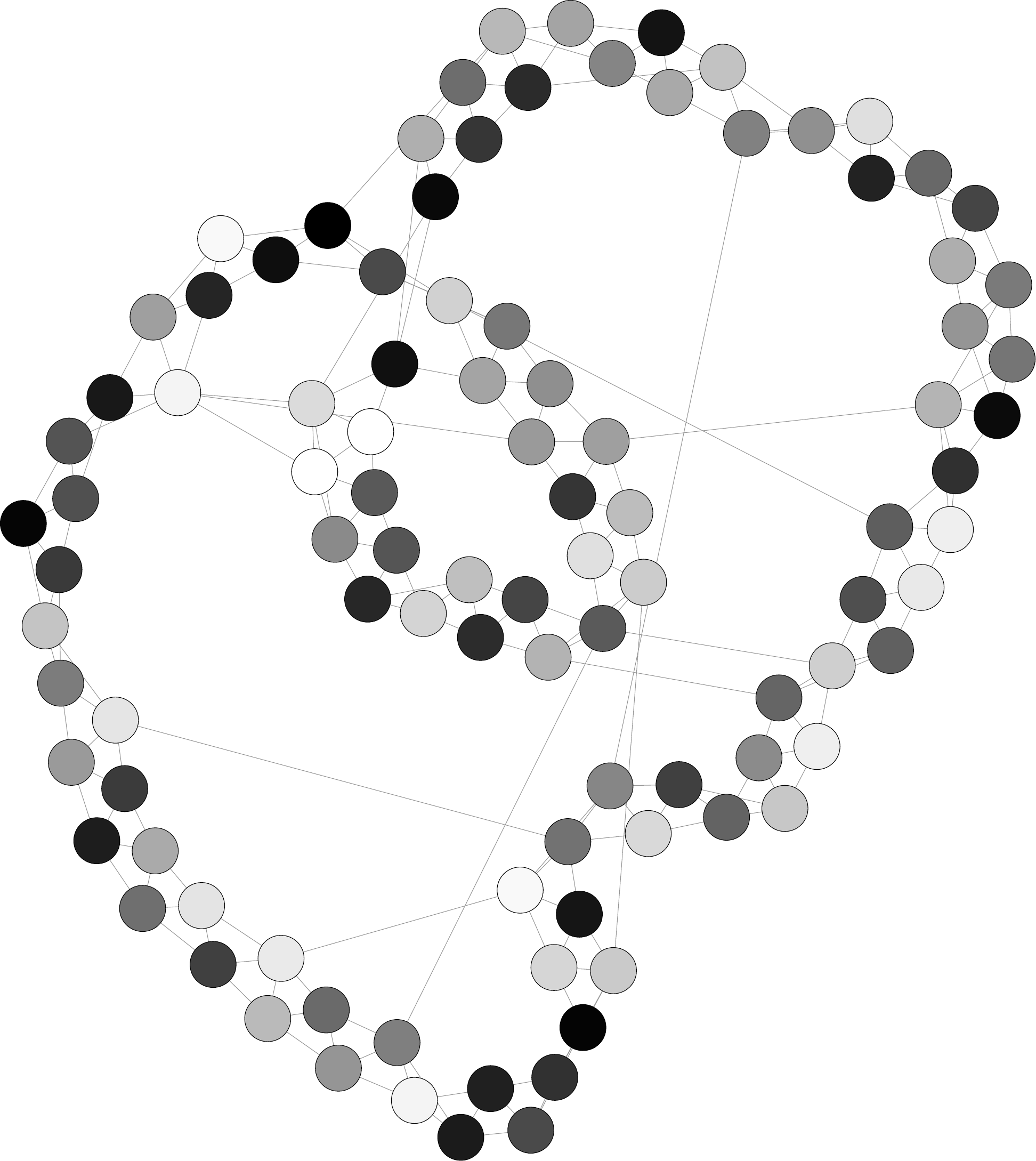} \hspace{1cm}
	\includegraphics[width=0.3\columnwidth]{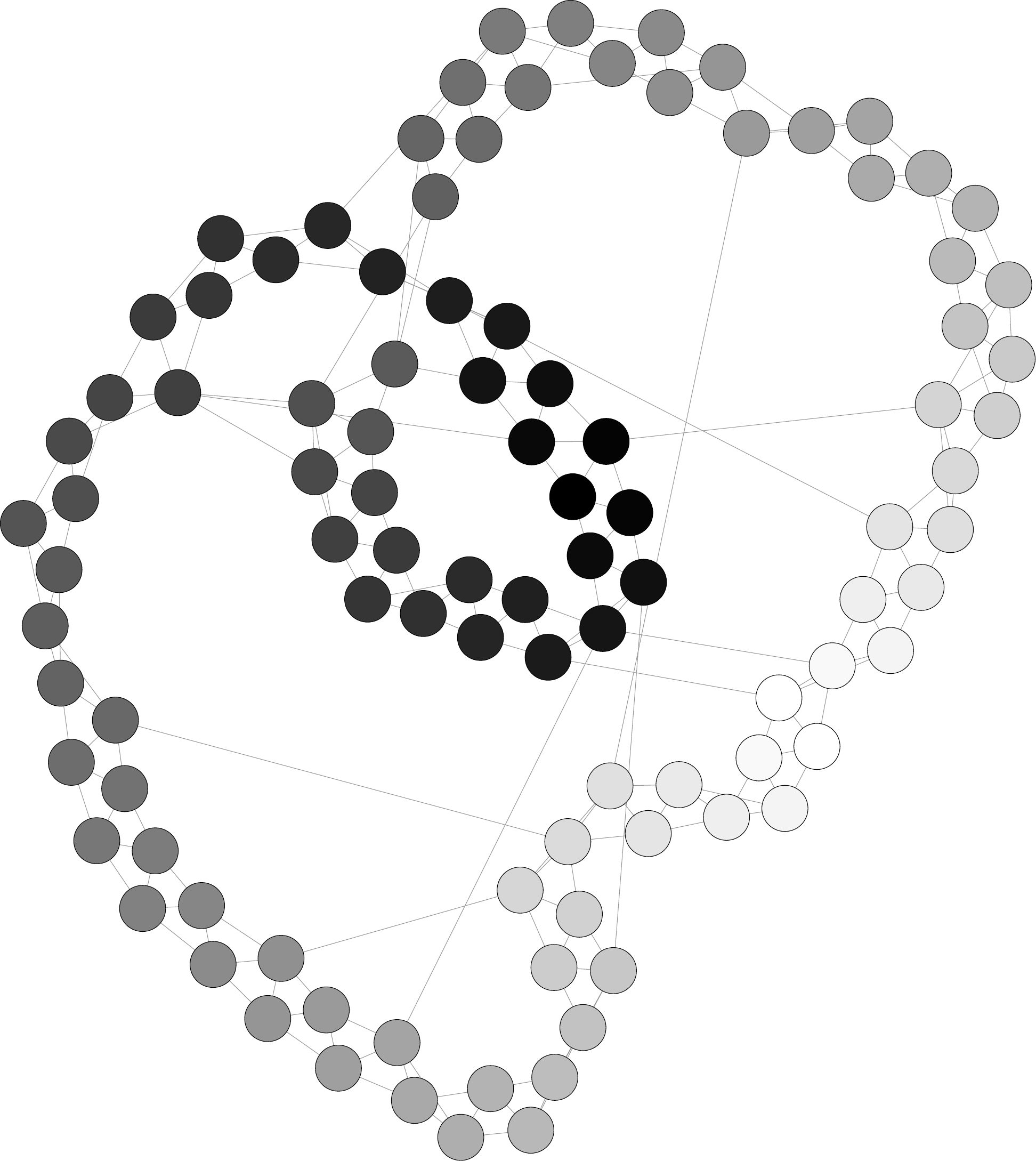}

	\caption{\label{fig:sw}The left and right figures show the vertices of a 
small-world network respectively before and after applying the heuristic \algo{} 
to relabel vertices. The color of the vertex depends on the label: two close 
colors means that the labels are close as well.\\ \cbs{} value before 
relabeling: 5283 ; \cbs{} value after relabeling: 752.}
\end{figure}

\begin{figure}[htp]
	\centering
	\includegraphics[width=0.35\columnwidth]{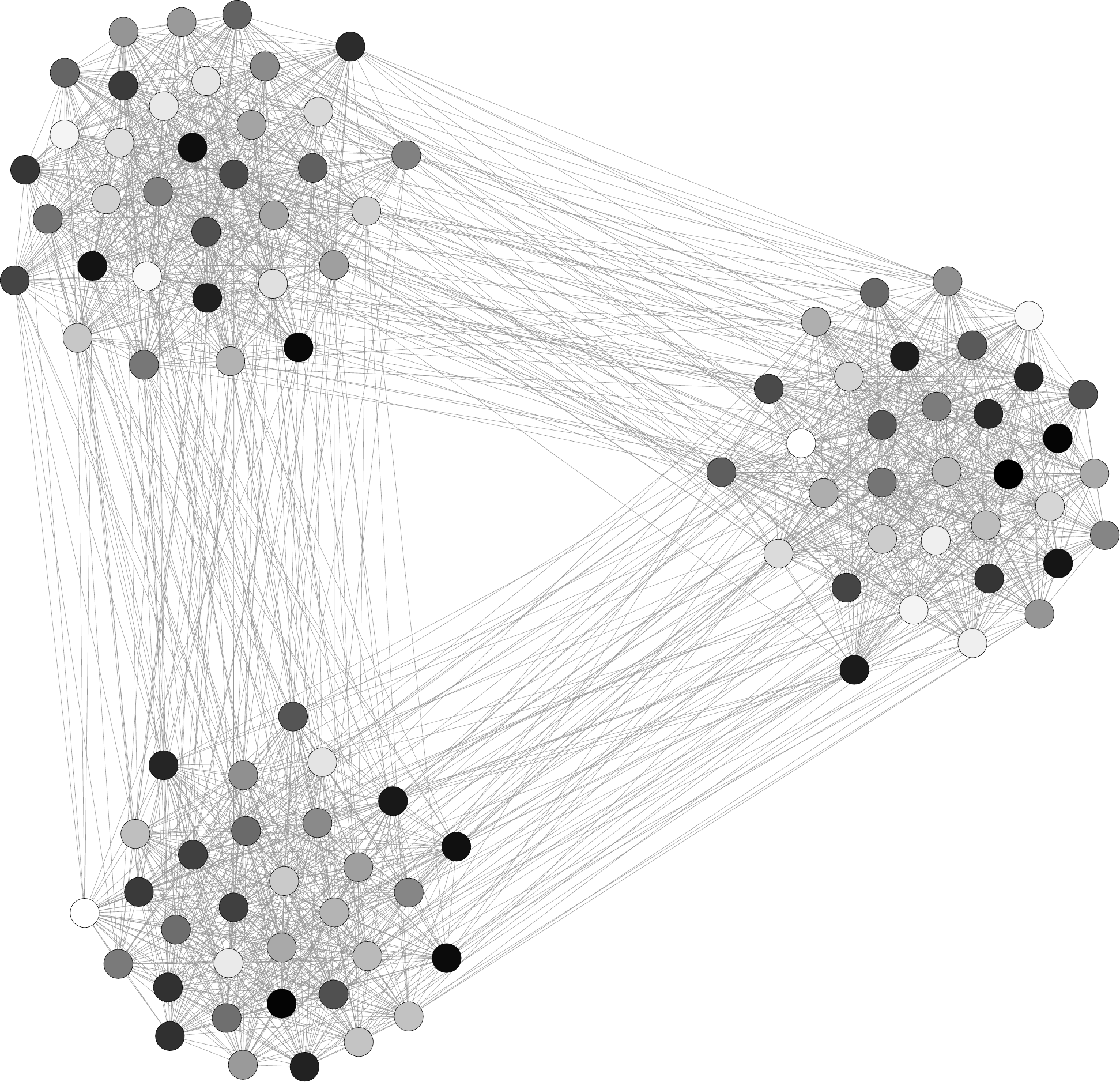} \hspace{1cm}
	\includegraphics[width=0.35\columnwidth]{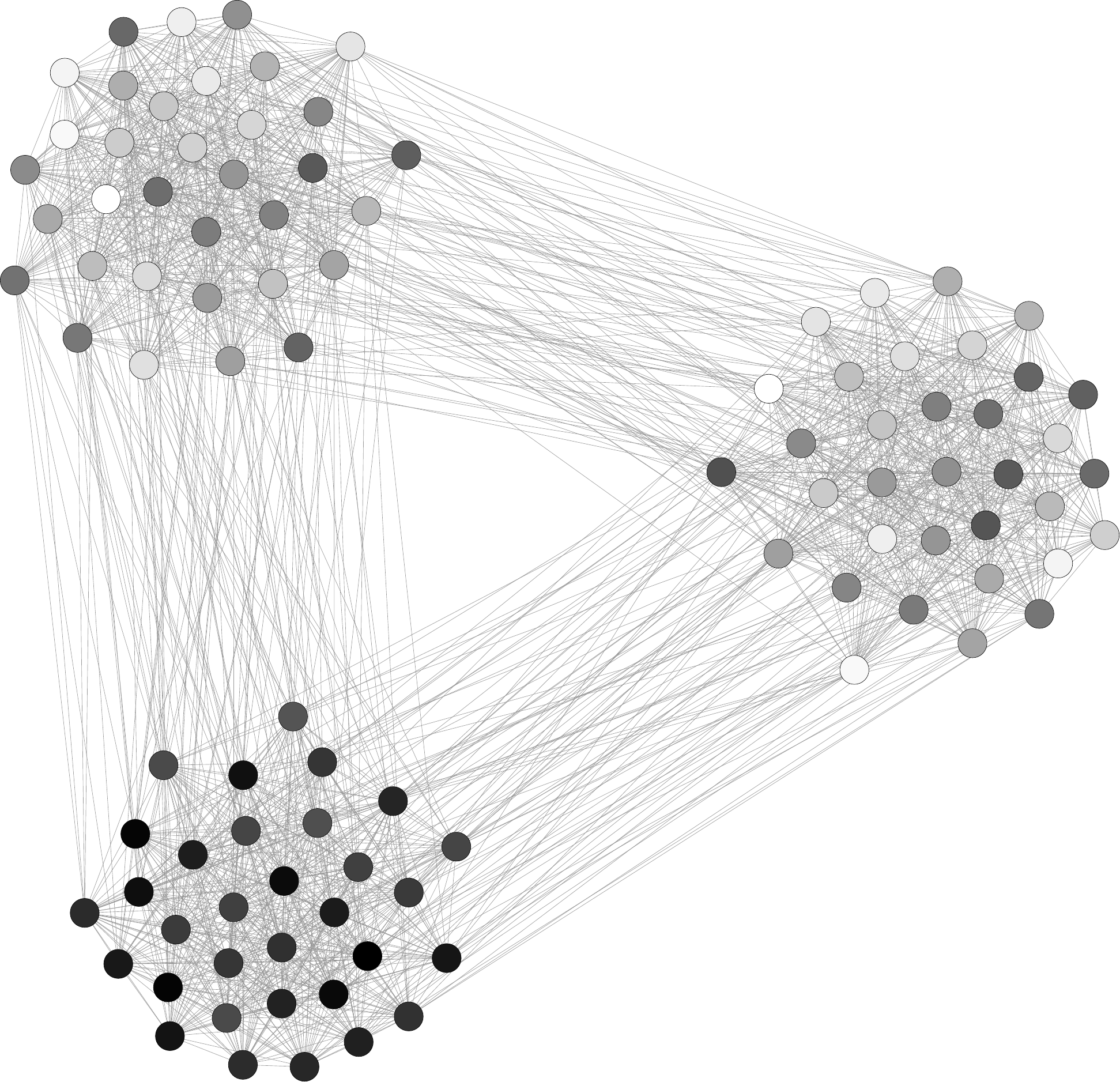}
	
	\caption{\label{fig:com}The left and right figures show the vertices of the 
network with communities respectively before and after applying the heuristic 
\algo{} to relabel vertices. The color of the vertex depends on the label: two 
close colors means that the labels are close as well. \\ \cbs{} value before 
relabeling: 45270 ; \cbs{} value after relabeling: 26589.}
\end{figure}

These illustrations show that the three types of networks highlight peculiar 
structures: the scale-free network has a stringy structure, the small-world has 
regular structures with random links and the structure in communities present 
three dense parts. The labels before relabeling do not follow in any sense this 
structure. The heuristic corrects this by relabeling the vertices in accordance 
with the topology. 
\makeatletter{}\section{Extension to weighted graphs}
\label{sec:extension}

Until now, the analysis was focused on unweighted graphs as up to our knowledge, 
there is no theoretical study about the cyclic bandwidth sum problem when the 
graph is considered as weighted. It is hence easier to assess the performance of 
the heuristic to restrain our study on that category of graphs. It is relevant 
nevertheless to consider the weight on graphs, as they are very common, 
especially in real-world graph analysis. The problem of \cbsp{} can be extended 
to take into account the weight of each edge in the sum of difference of labels. 
If we note $w_{uv}$ the weight between adjacent vertices $u$ and $v$, the 
weighted \cbsp{} is defined by:
\begin{align}
\min_{\mathbf{\pi}} f(\pi)\quad \text{with} \quad f(\mathbf{\pi})= 
\sum_{\{u,v\}\in E}w_{uv} d_H(\pi(u),\pi(v))
\end{align}

The weighted version of the heuristic is very similar to the one in the 
unweighted case. Two minor modifications have to be considered: the computation 
of the Jaccard similarity index in Step 1, as the neighborhood is influenced by 
the 
weights, and the incremental computation of \cbs{} in Step 2. The first problem 
can be addressed by defining a weighted similarity index between two vertices 
$u$ and $v$ as the following: 
\begin{align}
J_w(u,v) = \frac{N(u,v)}{D(u,v)}
\end{align}
where $N(u,v)$ represents the weight of neighbors shared by the two vertices and 
is defined by:
\begin{align}
N(u,v) =  2w_{u,v} + \sum_{\substack{x \in V \\ x\sim u \\ 
x\sim v}} {\min(w_{ux}, w_{vx}})
\end{align} 
and $D(u,v)$ represents the total weight of neighborhood of $u$ and $v$, and is 
defined by:
\begin{align}
D(u,v) =  2w_{u,v} + \sum_{\substack{x \in V \\ x\sim u \\ x\sim v}} 
{\frac{w_{ux} + w_{vw}}{2}} +  \sum_{\substack{x \in V \\ x\sim u \\ x\nsim 
v}}w_{ux} +   \sum_{\substack{x \in V \\ x\nsim u \\ x\sim v}}w_{vx}
\end{align}
We can note that if all weights are set to $1$ i.e. the graph is unweighted, we 
have $N(u,v) = 2 + \sharp\{\text{Common neighbors of $u$ and $v$}\}$ and $D = 
2 + \sharp\{\text{All neighbors of $u$ and $v$}\}$ which corresponds to the 
similarity index defined previously.

Adaptation of the incremental \cbs{} is, for its part, trivial, since it is only 
necessary to multiply each term we add and remove by the weight of the 
considered edge. 

There is no theoretical study about the weighted cyclic bandwidth sum problem to 
test the validity of the heuristic in the weighted case. Besides, it could be 
quite tricky to clearly characterize the structure of a weighted graph. We 
propose here only to illustrate the good behavior of the algorithm, by applying 
the heuristic on a weighted real-world complex network. The network is collected 
from data sets of the SocioPatterns projects \cite{Fournet2014}: it represents 
a face-to-face 
networks between students and teachers in a primary, where the people are the 
vertices of the graph and the weight between the vertices measures the 
cumulative duration of contact for one day.

Fig.~\ref{fig:school} shows the network after relabeling using the same color 
code as in the previous section. The layout is given by the data.

\begin{figure}[htp]
  \centering
	\includegraphics[width=0.60\columnwidth]{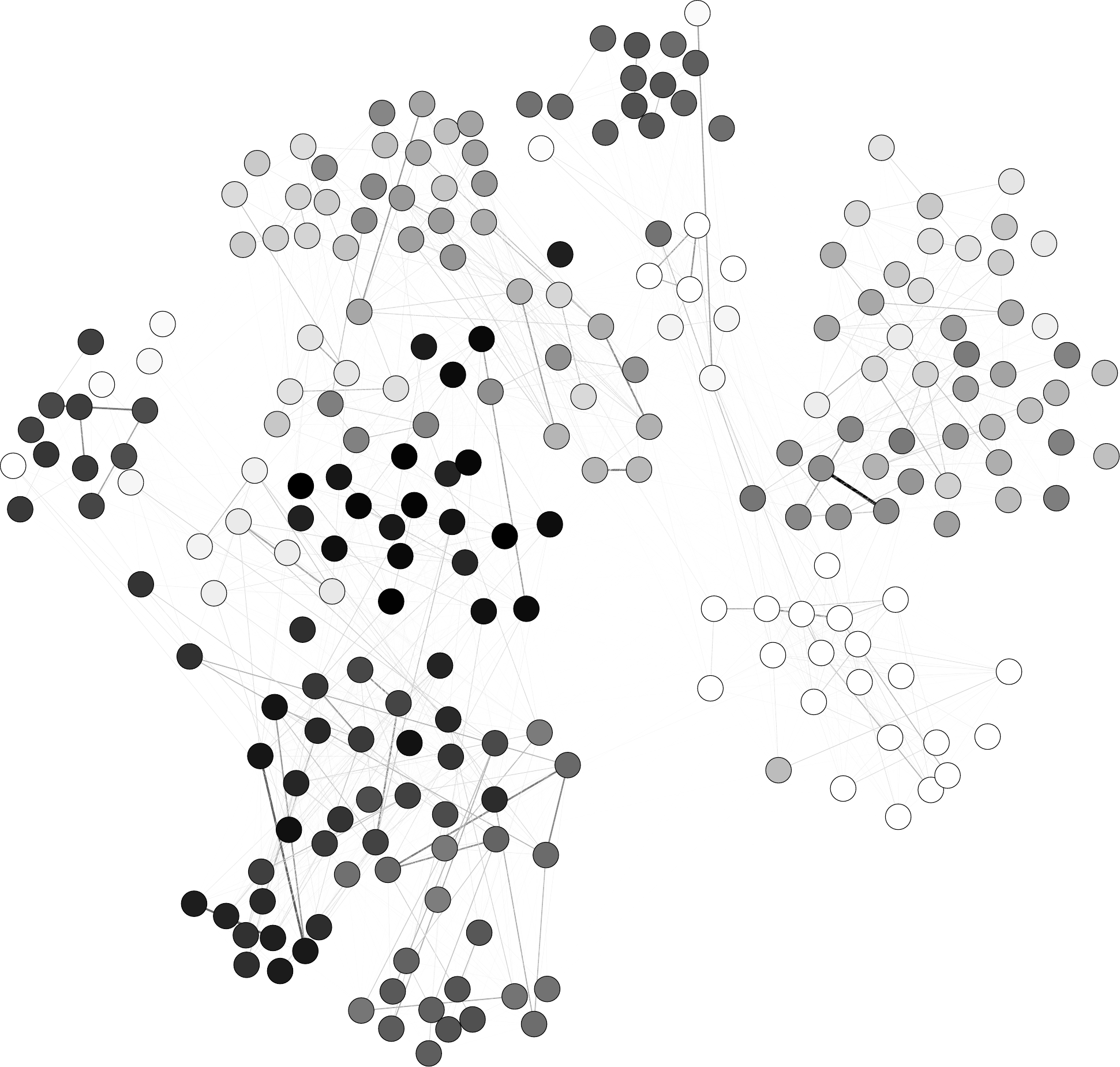}

\caption{\label{fig:school}Visualization of the face-to-face network in a 
primary school from the SocioPatterns project \cite{Fournet2014}. Each vertex 
represents either a student or a teacher. The color of the vertex depends on the 
label: two close colors means that the labels are close as well. The position 
of the vertices is from the data and reflects the topology of the network. The 
thickness of the edge denotes its weights.}
\end{figure}

The relation between the labeling and the topology is less obvious than 
previously. It stays nevertheless consistent with the positions of the vertices: 
the colors in the networks looks homogeneous. Each pile of vertices corresponds 
to a class in the school, and the vertices inside each class have close colors. 
Even if guarantees about the good behavior of the heuristic are not available, 
this example shows that the heuristic produces a reasonable result, in line with 
the ground truth, when it deals with weighted graphs.

\makeatletter{}\section{Conclusion}
\label{sec:conclusion}

The topology of a network is a crucial element in the analysis of processes 
lying over it. If the structure of networks can be described by models, 
it remains in many cases delicate to easily obtain an ordering of the vertices. 
We proposed in this paper a heuristic to discover the topology of the network 
without any assumptions on its structure. This heuristic has been of great 
interest to find an approximate solution of a known labeling problem called 
cyclic bandwidth sum problem, which has been used as a criterion to check the 
consistence between the topology and the labeling. Many extensions of 
this algorithm can be considered: we considered briefly the weighted graphs in 
the end of the paper. One could also deal with the case of directed graph, 
using the same idea of taking into account the local and global structure of 
the network.

\section*{Funding}

This work is supported by the programs ARC 5 and ARC 6 of the r\'egion 
Rh\^one-Alpes.

\clearpage

\makeatletter{}\appendix

\section{Incremental \cbs}
\label{apdx:theorems}

This appendix presents the results of the changes of the \cbs{} value for the 
incremental update of \cbs{} introduced in Section~\ref{subsec:step2}. The 
proofs of the following models follow the same reasoning as the proof of 
Theorem~\ref{thm:k01}.

\begin{theorem}{Edges between $k$ and the vertices of $O_2$}

Let $u \in O_2$ and $\Delta = \pi_i[u] - \pi_i[k]$. We have:
\begin{enumerate}
\item if $\Delta \leq \frac{n}{2} - p$ then
  $\cbs{}^{(i+1)}(k,u)=\cbs{}^{(i)}(k,u)+p$.
\item if $\Delta \geq \frac{n}{2} $ then
  $\cbs{}^{(i+1)}(k,u)=\cbs{}^{(i)}(k,u)-p$.
\item  if $\frac{n}{2} -p < \Delta < \frac{n}{2} $ then
$\cbs{}^{(i+1)}(k,u)=\cbs{}^{(i)}(k,u)-2\Delta + (n-p)$
\end{enumerate}
\end{theorem}

\begin{theorem}{Edges between $k$ and the vertices of $P$}

Let $u \in P$ and $\Delta = \pi_i[k] - \pi_i[u]$. We have:
\begin{enumerate}
\item if $(p+1) - \frac{n}{2} \leq \Delta \leq \frac{n}{2}$ then
  $\cbs{}^{(i+1)}(k,u)=\cbs{}^{(i)}(k,u)-2\Delta + (p+1)$.
\item if $\Delta > \frac{n}{2} $ then
  $\cbs{}^{(i+1)}(k,u)=\cbs{}^{(i)}(k,u)- n + (p+1)$.
\item  if $\Delta < (p+1) - \frac{n}{2} $ then
$\cbs{}^{(i+1)}(k,u)=\cbs{}^{(i)}(k,u)+ n - (p+1)$
\end{enumerate}
\end{theorem}

\begin{theorem}{Edges between the vertices of $P$ and the vertices of $O_1$}

Let $u \in P$, $v \in O_1$ and $\Delta = \pi_i[u] - \pi_i[v]$. We have:
\begin{enumerate}
\item if $\Delta \leq \frac{n}{2} - 1$ then
  $\cbs{}^{(i+1)}(u,v)=\cbs{}^{(i)}(u,v) + 1$.
\item if $\Delta \geq \frac{n}{2} $ then
  $\cbs{}^{(i+1)}(u,v)=\cbs{}^{(i)}(u,v) - 1$.
\item  if $\frac{n}{2} - 1 < \Delta < \frac{n}{2} $ then
$\cbs{}^{(i+1)}(u,v)=\cbs{}^{(i)}(u,v)$
\end{enumerate}
\end{theorem}

\begin{theorem}{Edges between the vertices of $P$ and the vertices of $O_2$}

Let $u \in P$, $v \in O_2$ and $\Delta = \pi_i[v] - \pi_i[u]$. We have:
\begin{enumerate}
\item if $\Delta \leq \frac{n}{2}$ then
  $\cbs{}^{(i+1)}(u,v)=\cbs{}^{(i)}(u,v) - 1$.
\item if $\Delta \geq \frac{n}{2} + 1 $ then
  $\cbs{}^{(i+1)}(u,v)=\cbs{}^{(i)}(u,v) + 1$.
\item  if $\frac{n}{2} < \Delta < \frac{n}{2} + 1 $ then
$\cbs{}^{(i+1)}(u,v)=\cbs{}^{(i)}(u,v)$
\end{enumerate}
\end{theorem}

\section{Detailed results of computational experiments}
\label{sec:detailed_results}

This section presents the results for the Cartesian products of graphs described 
in Section~\ref{sec:results}, for different values of $m$ and $n$.

\subsection{Performance}

For the following tables, the first two columns give the name of the graphs 
used in the Cartesian products, the third column displays the median values of 
the distribution of the \cbs{} over 30 repetitions, the fourth column the 
median absolute deviation of the distribution of the \cbs{} value over 30 
repetition and the fifth column the theoretical upper bound. Finally the sixth 
column displays the mean execution time in seconds for one repetition of the 
heuristic.

\begin{table}[!ht]
	\centering
	\scriptsize
	\begin{tabular}{|ll|rr|r|r|r|r|} 
	\hline
	\multicolumn{1}{|c}{$G$} &
	\multicolumn{1}{c|}{$H$} &
	\multicolumn{1}{c}{\emph{median \cbs{}}} &
	\multicolumn{1}{c|}{\emph{mad \cbs{}}} &
	\multicolumn{1}{c|}{\emph{ub}} &
	\multicolumn{1}{c|}{\emph{rd}} &
	\multicolumn{1}{c|}{\emph{time (s)}} \\ \hline
		
	\makeatletter{}$P_{5}$ & $K_{5}$ & 200 & 0 & 395 & -0.49 & 0.01 \\ 
$P_{5}$ & $K_{10}$ & 1225 & 0 & 3165 & -0.61 & 0.01 \\ 
$P_{5}$ & $K_{15}$ & 3700 & 0 & 10560 & -0.65 & 0.03 \\ 
$P_{5}$ & $K_{20}$ & 8250 & 0 & 25080 & -0.67 & 0.09 \\ 
$P_{10}$ & $K_{5}$ & 425 & 0 & 1545 & -0.72 & 0.01 \\ 
$P_{10}$ & $K_{10}$ & 2550 & 0 & 12590 & -0.80 & 0.05 \\ 
$P_{10}$ & $K_{15}$ & 7625 & 0 & 42135 & -0.82 & 0.08 \\ 
$P_{10}$ & $K_{20}$ & 16900 & 0 & 100180 & -0.83 & 0.15 \\ 
$P_{15}$ & $K_{5}$ & 650 & 0 & 3445 & -0.81 & 0.01 \\ 
$P_{15}$ & $K_{10}$ & 3875 & 0 & 28265 & -0.86 & 0.05 \\ 
$P_{15}$ & $K_{15}$ & 11550 & 0 & 94710 & -0.88 & 0.13 \\ 
$P_{15}$ & $K_{20}$ & 25550 & 0 & 225280 & -0.89 & 0.27 \\ 
$P_{20}$ & $K_{5}$ & 875 & 0 & 6095 & -0.86 & 0.03 \\ 
$P_{20}$ & $K_{10}$ & 5200 & 0 & 50190 & -0.90 & 0.08 \\ 
$P_{20}$ & $K_{15}$ & 15475 & 0 & 168285 & -0.91 & 0.13 \\ 
$P_{20}$ & $K_{20}$ & 34200 & 0 & 400380 & -0.91 & 0.31 \\

	\hline
	\end{tabular}
	\caption{\label{tab:cart_pk_perf}Results for the Cartesian products of a 
path and a complete graph}

\end{table}

\begin{table}[!ht]
	\centering
	\scriptsize
	\begin{tabular}{|ll|rr|r|r|r|r|} 
	\hline
	\multicolumn{1}{|c}{$G$} &
	\multicolumn{1}{c|}{$H$} &
	\multicolumn{1}{c}{\emph{median \cbs{}}} &
	\multicolumn{1}{c|}{\emph{mad \cbs{}}} &
	\multicolumn{1}{c|}{\emph{ub}} &
	\multicolumn{1}{c|}{\emph{rd}} &
	\multicolumn{1}{c|}{\emph{time (s)}} \\ \hline
		
	\makeatletter{}$C_{5}$ & $K_{5}$ & 225 & 0 & 415 & -0.46 & $<0.01$ \\ 
$C_{5}$ & $K_{10}$ & 1325 & 0 & 3205 & -0.59 & 0.04 \\ 
$C_{5}$ & $K_{15}$ & 3925 & 0 & 10620 & -0.63 & 0.05 \\ 
$C_{5}$ & $K_{20}$ & 8650 & 0 & 25160 & -0.66 & 0.07 \\ 
$C_{10}$ & $K_{5}$ & 450 & 0 & 1590 & -0.72 & 0.02 \\ 
$C_{10}$ & $K_{10}$ & 2650 & 0 & 12680 & -0.79 & 0.02 \\ 
$C_{10}$ & $K_{15}$ & 7850 & 0 & 42270 & -0.81 & 0.07 \\ 
$C_{10}$ & $K_{20}$ & 17300 & 0 & 100360 & -0.83 & 0.15 \\ 
$C_{15}$ & $K_{5}$ & 675 & 0 & 3515 & -0.81 & 0.01 \\ 
$C_{15}$ & $K_{10}$ & 3975 & 0 & 28405 & -0.86 & 0.07 \\ 
$C_{15}$ & $K_{15}$ & 11775 & 0 & 94920 & -0.88 & 0.16 \\ 
$C_{15}$ & $K_{20}$ & 25950 & 0 & 225560 & -0.88 & 0.25 \\ 
$C_{20}$ & $K_{5}$ & 900 & 0 & 6190 & -0.85 & 0.03 \\ 
$C_{20}$ & $K_{10}$ & 5300 & 0 & 50380 & -0.89 & 0.1 \\ 
$C_{20}$ & $K_{15}$ & 15700 & 0 & 168570 & -0.91 & 0.19 \\ 
$C_{20}$ & $K_{20}$ & 34600 & 0 & 400760 & -0.91 & 0.24 \\

	\hline
	\end{tabular}
	\caption{\label{tab:cart_ck_perf}Results for the Cartesian products of a 
cycle and a complete 
graph}

\end{table}

\begin{table}[!ht]
	\centering
	\scriptsize
	\begin{tabular}{|ll|rr|r|r|r|r|} 
	\hline
	\multicolumn{1}{|c}{$G$} &
	\multicolumn{1}{c|}{$H$} &
	\multicolumn{1}{c}{\emph{median \cbs{}}} &
	\multicolumn{1}{c|}{\emph{mad \cbs{}}} &
	\multicolumn{1}{c|}{\emph{ub}} &
	\multicolumn{1}{c|}{\emph{rd}} &
	\multicolumn{1}{c|}{\emph{time (s)}} \\ \hline
		
	\makeatletter{}$K_{5}$ & $K_{5}$ & 547 & 20 & 474 & 0.15 & 0.01 \\ 
$K_{10}$ & $K_{5}$ & 2305 & 2 & 3324 & -0.31 & 0.02 \\ 
$K_{10}$ & $K_{10}$ & 21349 & 466 & 14149 & 0.51 & 0.2 \\ 
$K_{15}$ & $K_{5}$ & 6130 & 7 & 10800 & -0.43 & 0.06 \\ 
$K_{15}$ & $K_{10}$ & 33347 & 18 & 44475 & -0.25 & 0.17 \\ 
$K_{15}$ & $K_{15}$ & 169974 & 1766 & 102900 & 0.65 & 1.07 \\ 
$K_{20}$ & $K_{5}$ & 12571 & 9 & 25399 & -0.51 & 0.1 \\ 
$K_{20}$ & $K_{10}$ & 62644 & 31 & 103299 & -0.39 & 0.23 \\ 
$K_{20}$ & $K_{15}$ & 187710 & 13 & 236199 & -0.21 & 0.47 \\ 
$K_{20}$ & $K_{20}$ & 740458 & 3582 & 426599 & 0.74 & 3.46 \\

	\hline
	\end{tabular}
	\caption{\label{tab:cart_kk_perf}Results for the Cartesian products of 
complete graphs}

\end{table}

\begin{table}[!ht]
	\centering
	\scriptsize
	\begin{tabular}{|ll|rr|r|r|r|r|} 
	\hline
	\multicolumn{1}{|c}{$G$} &
	\multicolumn{1}{c|}{$H$} &
	\multicolumn{1}{c}{\emph{median \cbs{}}} &
	\multicolumn{1}{c|}{\emph{mad \cbs{}}} &
	\multicolumn{1}{c|}{\emph{ub}} &
	\multicolumn{1}{c|}{\emph{rd}} &
	\multicolumn{1}{c|}{\emph{time (s)}} \\ \hline
		
	\makeatletter{}$P_{5}$ & $C_{5}$ & 158 & 7 & 145 & 0.09 & 0.01 \\ 
$P_{5}$ & $C_{10}$ & 481 & 28 & 290 & 0.66 & 0.05 \\ 
$P_{5}$ & $C_{15}$ & 920 & 50 & 435 & 1.11 & 0.06 \\ 
$P_{5}$ & $C_{20}$ & 1508 & 108 & 580 & 1.60 & 0.1 \\ 
$P_{10}$ & $C_{5}$ & 448 & 47 & 545 & -0.18 & 0.02 \\ 
$P_{10}$ & $C_{10}$ & 1472 & 96 & 1090 & 0.35 & 0.12 \\ 
$P_{10}$ & $C_{15}$ & 3031 & 218 & 1635 & 0.85 & 0.33 \\ 
$P_{10}$ & $C_{20}$ & 4493 & 315 & 2180 & 1.06 & 0.65 \\ 
$P_{15}$ & $C_{5}$ & 739 & 112 & 1195 & -0.38 & 0.09 \\ 
$P_{15}$ & $C_{10}$ & 2680 & 108 & 2390 & 0.12 & 0.35 \\ 
$P_{15}$ & $C_{15}$ & 5464 & 338 & 3585 & 0.52 & 0.94 \\ 
$P_{15}$ & $C_{20}$ & 8839 & 410 & 4780 & 0.85 & 1.95 \\ 
$P_{20}$ & $C_{5}$ & 1000 & 94 & 2095 & -0.52 & 0.09 \\ 
$P_{20}$ & $C_{10}$ & 3751 & 333 & 4190 & -0.10 & 0.48 \\ 
$P_{20}$ & $C_{15}$ & 8105 & 519 & 6285 & 0.29 & 1.28 \\ 
$P_{20}$ & $C_{20}$ & 13902 & 816 & 8380 & 0.66 & 3.74 \\

	\hline
	\end{tabular}
	\caption{\label{tab:cart_pc_perf}Results for the Cartesian products of a 
path and a cycle}

\end{table}

\begin{table}[!ht]
	\centering
	\scriptsize
	\begin{tabular}{|ll|rr|r|r|r|r|} 
	\hline
	\multicolumn{1}{|c}{$G$} &
	\multicolumn{1}{c|}{$H$} &
	\multicolumn{1}{c}{\emph{median \cbs{}}} &
	\multicolumn{1}{c|}{\emph{mad \cbs{}}} &
	\multicolumn{1}{c|}{\emph{ub}} &
	\multicolumn{1}{c|}{\emph{rd}} &
	\multicolumn{1}{c|}{\emph{time (s)}} \\ \hline
		
	\makeatletter{}$P_{5}$ & $P_{5}$ & 142 & 13 & 120 & 0.18 & $<0.01$ \\ 
$P_{10}$ & $P_{5}$ & 426 & 47 & 265 & 0.61 & 0.03 \\ 
$P_{10}$ & $P_{10}$ & 1397 & 85 & 990 & 0.41 & 0.08 \\ 
$P_{15}$ & $P_{5}$ & 759 & 78 & 410 & 0.85 & 0.05 \\ 
$P_{15}$ & $P_{10}$ & 2671 & 305 & 1535 & 0.74 & 0.31 \\ 
$P_{15}$ & $P_{15}$ & 5199 & 621 & 3360 & 0.55 & 0.6 \\ 
$P_{20}$ & $P_{5}$ & 1164 & 77 & 555 & 1.10 & 0.12 \\ 
$P_{20}$ & $P_{10}$ & 3882 & 320 & 2080 & 0.87 & 0.41 \\ 
$P_{20}$ & $P_{15}$ & 7734 & 791 & 4555 & 0.70 & 1.49 \\ 
$P_{20}$ & $P_{20}$ & 11858 & 1670 & 7980 & 0.49 & 3.42 \\

	\hline
	\end{tabular}
	\caption{\label{tab:cart_pp_perf}Results for the Cartesian products of 
paths}

\end{table}

\begin{table}[!ht]
	\centering
	\scriptsize
	\begin{tabular}{|ll|rr|r|r|r|r|} 
	\hline
	\multicolumn{1}{|c}{$G$} &
	\multicolumn{1}{c|}{$H$} &
	\multicolumn{1}{c}{\emph{median \cbs{}}} &
	\multicolumn{1}{c|}{\emph{mad \cbs{}}} &
	\multicolumn{1}{c|}{\emph{ub}} &
	\multicolumn{1}{c|}{\emph{rd}} &
	\multicolumn{1}{c|}{\emph{time (s)}} \\ \hline
		
	\makeatletter{}$C_{5}$ & $C_{5}$ & 194 & 7 & 165 & 0.18 & 0.01 \\ 
$C_{10}$ & $C_{5}$ & 488 & 52 & 330 & 0.48 & 0.04 \\ 
$C_{10}$ & $C_{10}$ & 1821 & 68 & 1180 & 0.54 & 0.18 \\ 
$C_{15}$ & $C_{5}$ & 853 & 139 & 495 & 0.72 & 0.11 \\ 
$C_{15}$ & $C_{10}$ & 3376 & 180 & 1770 & 0.91 & 0.42 \\ 
$C_{15}$ & $C_{15}$ & 6610 & 357 & 3795 & 0.74 & 0.79 \\ 
$C_{20}$ & $C_{5}$ & 1086 & 172 & 660 & 0.65 & 0.15 \\ 
$C_{20}$ & $C_{10}$ & 4606 & 545 & 2360 & 0.95 & 0.84 \\ 
$C_{20}$ & $C_{15}$ & 9737 & 728 & 5060 & 0.92 & 2.06 \\ 
$C_{20}$ & $C_{20}$ & 16637 & 705 & 8760 & 0.90 & 3.65 \\

	\hline
	\end{tabular}
	\caption{\label{tab:cart_cc_perf}Results for the Cartesian products of 
cycles}

\end{table}

\clearpage

For the following tables, the first three columns give the name of the 
graph if the graph is from the Harwell-Boeing collection, of the value of 
$p$ is the graph is a random graph, and the number of vertices 
(\emph{$\#V$}) and the number of edges (\emph{$\#E$}) of the graph. The next 
three columns give the results for the heuristic \algo{}, with the median value 
of \cbs{} value over 30 repetitions (\emph{median \cbs{}}), the median absolute 
deviation (\emph{mad \cbs{}}) and the execution time for one execution in 
seconds. (\emph{time}). The next three column five the results for the heuristic 
\gvns{}, with the \cbs{} value obtained without randomizing initial vertex 
ordering (\emph{\cbs{} w/o r}), the \cbs{} value obtained with randomizing 
initial vertex ordering (\emph{\cbs{} w/ r}) and the execution time in seconds 
(\emph{time}). Finally the last column gives the relative distance between 
\emph{median \cbs{}} and \emph{\cbs{} 
w/ r}.
\begin{table}[!ht]
	\centering
	\scriptsize
	\begin{tabular}{|lll|rrr|rrr|r|}
	\hline
\multicolumn{3}{|c}{\textbf{Graph}} &  
\multicolumn{3}{|c|}{\textbf{\algo{}}} & 
\multicolumn{3}{c|}{\textbf{\gvns{}}} & \\ \hline

\multicolumn{1}{|c}{Name} & 
\multicolumn{1}{c}{$\#V$} & 
\multicolumn{1}{c|}{$\#E$} &

\multicolumn{1}{c}{\emph{median \cbs{}}} & 
\multicolumn{1}{c}{\emph{mad \cbs{}}} & 
\multicolumn{1}{c|}{ \emph{time (s)}} &

\multicolumn{1}{c}{\emph{\cbs{} w/o r}} &
\multicolumn{1}{c}{\emph{\cbs{} w/ r}} &
\multicolumn{1}{c|}{ \emph{time (s)}} &
\multicolumn{1}{c|}{\emph{rd}} \\ 
\hline

\makeatletter{}bcspwr01 & 39 & 46 & 106 & 3 & 0.01 & 212 & 241 & 104.0 & -0.56 \\ 
 bcspwr02 & 49 & 59 & 164 & 3 & 0.02 & 350 & 388 & 114.0 & -0.58 \\ 
 bcspwr03 & 118 & 179 & 850 & 64 & 0.23 & 2035 & 3585 & 270.0 & -0.76 \\ 
 bcspwr04 & 274 & 669 & 6280 & 322 & 1.5 & 34135 & 36627 & 1060.0 & -0.83 \\ 
 bcspwr05 & 443 & 590 & 6032 & 440 & 3.78 & 41194 & 51486 & 1590.0 & -0.88 \\ 
 bcspwr06 & 1454 & 1923 & 38434 & 3621 & 89.0 & 84704 & 622141 & 8620.0 & -0.94 \\ 
 dwt59 & 59 & 104 & 322 & 29 & 0.03 & 545 & 923 & 189.0 & -0.65 \\ 
 dwt72 & 72 & 75 & 204 & 10 & 0.08 & 180 & 744 & 127.0 & -0.73 \\ 
 dwt87 & 87 & 227 & 1118 & 86 & 0.11 & 2785 & 3339 & 392.0 & -0.67 \\ 
 dwt162 & 162 & 510 & 3235 & 669 & 0.35 & 5779 & 15628 & 349.0 & -0.79 \\ 
 dwt193 & 193 & 1650 & 30811 & 1322 & 0.77 & 46150 & 66408 & 1590.0 & -0.54 \\ 
 dwt221 & 221 & 704 & 6633 & 642 & 0.7 & 15414 & 30433 & 972.0 & -0.78 \\ 
 dwt419 & 419 & 1572 & 24779 & 1231 & 3.8 & 84803 & 141523 & 3510.0 & -0.82 \\ 
 dwt592 & 592 & 2256 & 43185 & 2119 & 9.47 & 72302 & 292871 & 3700.0 & -0.85 \\ 
 dwt992 & 992 & 7876 & 286660 & 14260 & 34.1 & 800566 & 1827354 & 23800.0 & -0.84 \\ 
 can24 & 24 & 68 & 207 & 11 & 0.02 & 232 & 229 & 76.7 & -0.10 \\ 
 can61 & 61 & 248 & 1553 & 0 & 0.02 & 2385 & 2556 & 291.0 & -0.39 \\ 
 can62 & 62 & 78 & 247 & 15 & 0.03 & 389 & 713 & 125.0 & -0.65 \\ 
 can73 & 73 & 152 & 1003 & 43 & 0.12 & 1413 & 1838 & 155.0 & -0.45 \\ 
 can96 & 96 & 336 & 2512 & 145 & 0.12 & 3884 & 5535 & 237.0 & -0.55 \\ 
 can144 & 144 & 576 & 11204 & 94 & 0.06 & 12055 & 15342 & 490.0 & -0.27 \\ 
 can187 & 187 & 652 & 5363 & 775 & 0.57 & 14112 & 23658 & 1070.0 & -0.77 \\ 
 can229 & 229 & 774 & 11272 & 1191 & 0.83 & 18362 & 35230 & 972.0 & -0.68 \\ 
 can268 & 268 & 1407 & 36046 & 2867 & 1.03 & 44034 & 77680 & 1510.0 & -0.54 \\ 
 can715 & 715 & 2975 & 87773 & 4651 & 12.7 & 149144 & 473748 & 4830.0 & -0.81 \\ 
 can838 & 838 & 4586 & 373361 & 16932 & 15.4 & 547732 & 880386 & 6560.0 & -0.58 \\ 
 can1054 & 1054 & 5571 & 335935 & 30665 & 33.2 & 562451 & 1359525 & 12800.0 & -0.75 \\ 
  
\hline

\end{tabular}

	\caption{\label{tab:harwell_perf}Results for the graph from the 
Harwell-Boeing collection.}
\end{table}

\begin{table}[!ht]
	\centering
	\scriptsize
	\begin{tabular}{|lll|rrr|rrr|r|}
	\hline
\multicolumn{3}{|c}{\textbf{Graph}} &  \multicolumn{3}{|c|}{\textbf{\algo{}}} & 
\multicolumn{3}{c|}{\textbf{\gvns{}}} & \\ \hline

\multicolumn{1}{|c}{p} & \multicolumn{1}{c}{$\#V$} & 
\multicolumn{1}{c|}{$\#E$} &

\multicolumn{1}{|c}{\emph{median \cbs{}}} & \multicolumn{1}{c}{\emph{mad 
\cbs{}}} & \multicolumn{1}{c|}{ \emph{time (s)}} &

\multicolumn{1}{|c}{\emph{\cbs{} w/o r}} & \multicolumn{1}{c}{\emph{\cbs{} w/ 
r}} & \multicolumn{1}{c|}{ \emph{time (s)}} & 
\multicolumn{1}{c|}{\emph{rd}} \\ 
\hline

\makeatletter{}0.1 & 0 & 455 & 11236 & 0  & 0.32 & 9315 & 9355 & 413.0 & 0.20 \\ 
 0.1 & 1 & 524 & 13169 & 0  & 0.29 & 10600 & 10883 & 470.0 & 0.21 \\ 
 0.1 & 2 & 497 & 12682 & 0  & 0.31 & 10379 & 10423 & 376.0 & 0.22 \\ 
 0.1 & 3 & 501 & 13013 & 0  & 0.26 & 10194 & 10225 & 383.0 & 0.27 \\ 
 0.1 & 4 & 533 & 14011 & 0  & 0.3 & 10901 & 11191 & 408.0 & 0.25 \\ 
 0.1 & 5 & 488 & 12295 & 0  & 0.28 & 10105 & 10317 & 401.0 & 0.19 \\ 
 0.1 & 6 & 510 & 12772 & 0  & 0.28 & 10783 & 10303 & 520.0 & 0.24 \\ 
 0.1 & 7 & 487 & 12439 & 0  & 0.21 & 10081 & 10094 & 578.0 & 0.23 \\ 
 0.1 & 8 & 465 & 11654 & 0  & 0.14 & 9529 & 9401 & 578.0 & 0.24 \\ 
 0.1 & 9 & 517 & 13036 & 0  & 0.32 & 10703 & 10866 & 660.0 & 0.20 \\ 
 0.3 & 0 & 1500 & 37587 & 0  & 0.52 & 33331 & 34209 & 2110.0 & 0.10 \\ 
 0.3 & 1 & 1449 & 35784 & 0  & 0.24 & 32354 & 33373 & 2420.0 & 0.07 \\ 
 0.3 & 2 & 1494 & 38841 & 0  & 0.36 & 34174 & 34139 & 3120.0 & 0.14 \\ 
 0.3 & 3 & 1537 & 39000 & 0  & 0.16 & 35006 & 35281 & 3080.0 & 0.11 \\ 
 0.3 & 4 & 1541 & 39524 & 0  & 0.4 & 35246 & 35424 & 2660.0 & 0.12 \\ 
 0.3 & 5 & 1475 & 36670 & 0  & 0.35 & 33872 & 33975 & 2450.0 & 0.08 \\ 
 0.3 & 6 & 1463 & 36893 & 0  & 0.24 & 33488 & 32930 & 2390.0 & 0.12 \\ 
 0.3 & 7 & 1511 & 38280 & 0  & 0.25 & 34399 & 34551 & 2720.0 & 0.11 \\ 
 0.3 & 8 & 1482 & 38131 & 0  & 0.32 & 33813 & 33962 & 2770.0 & 0.12 \\ 
 0.3 & 9 & 1506 & 37388 & 0  & 0.41 & 34492 & 34374 & 3000.0 & 0.09 \\ 
 0.5 & 0 & 2406 & 60937 & 0  & 0.4 & 56308 & 55603 & 3880.0 & 0.10 \\ 
 0.5 & 1 & 2482 & 62238 & 0  & 0.32 & 57704 & 58516 & 6050.0 & 0.06 \\ 
 0.5 & 2 & 2446 & 61217 & 0  & 0.26 & 57293 & 56651 & 5290.0 & 0.08 \\ 
 0.5 & 3 & 2437 & 61659 & 0  & 0.42 & 57678 & 56448 & 5650.0 & 0.09 \\ 
 0.5 & 4 & 2471 & 62437 & 0  & 0.4 & 58010 & 57770 & 5200.0 & 0.08 \\ 
 0.5 & 5 & 2503 & 62970 & 0  & 0.7 & 58613 & 58754 & 4480.0 & 0.07 \\ 
 0.5 & 6 & 2500 & 63012 & 0  & 0.47 & 58612 & 58699 & 4800.0 & 0.07 \\ 
 0.5 & 7 & 2453 & 61796 & 0  & 0.58 & 57868 & 57506 & 2820.0 & 0.07 \\ 
 0.5 & 8 & 2468 & 61736 & 0  & 0.45 & 58727 & 58030 & 3590.0 & 0.06 \\ 
 0.5 & 9 & 2425 & 60957 & 0  & 0.42 & 57344 & 56866 & 3730.0 & 0.07 \\ 
 0.7 & 0 & 3468 & 87649 & 0  & 0.66 & 83380 & 84035 & 8580.0 & 0.04 \\ 
 0.7 & 1 & 3470 & 88077 & 0  & 0.46 & 82910 & 83215 & 9900.0 & 0.06 \\ 
 0.7 & 2 & 3432 & 86837 & 0  & 0.61 & 82959 & 82992 & 8300.0 & 0.05 \\ 
 0.7 & 3 & 3437 & 87006 & 0  & 0.66 & 83165 & 82268 & 7250.0 & 0.06 \\ 
 0.7 & 4 & 3450 & 87801 & 0  & 0.46 & 83799 & 82603 & 4770.0 & 0.06 \\ 
 0.7 & 5 & 3436 & 87372 & 0  & 0.48 & 83288 & 82933 & 5360.0 & 0.05 \\ 
 0.7 & 6 & 3473 & 88352 & 0  & 0.65 & 84383 & 83628 & 6700.0 & 0.06 \\ 
 0.7 & 7 & 3483 & 88322 & 0  & 0.46 & 84468 & 84277 & 8700.0 & 0.05 \\ 
 0.7 & 8 & 3443 & 87339 & 0  & 0.78 & 83285 & 81972 & 6640.0 & 0.07 \\ 
 0.7 & 9 & 3518 & 88014 & 0  & 0.64 & 84904 & 84211 & 6210.0 & 0.05 \\ 
 0.9 & 0 & 4428 & 111117 & 0  & 0.67 & 108761 & 109564 & 13200.0 & 0.01 \\ 
 0.9 & 1 & 4480 & 113297 & 0  & 0.67 & 110700 & 110148 & 12500.0 & 0.03 \\ 
 0.9 & 2 & 4459 & 113026 & 0  & 0.87 & 110413 & 110317 & 9540.0 & 0.02 \\ 
 0.9 & 3 & 4468 & 112714 & 0  & 0.92 & 110782 & 109953 & 6480.0 & 0.03 \\ 
 0.9 & 4 & 4469 & 113236 & 0  & 0.68 & 109824 & 109934 & 8890.0 & 0.03 \\ 
 0.9 & 5 & 4460 & 112747 & 0  & 0.67 & 110390 & 109940 & 11100.0 & 0.03 \\ 
 0.9 & 6 & 4433 & 111787 & 0  & 0.68 & 108919 & 108971 & 8780.0 & 0.03 \\ 
 0.9 & 7 & 4471 & 112802 & 0  & 0.88 & 110596 & 110817 & 8070.0 & 0.02 \\ 
 0.9 & 8 & 4457 & 112511 & 0  & 0.5 & 109174 & 109762 & 6560.0 & 0.03 \\ 
 0.9 & 9 & 4451 & 112342 & 0  & 0.66 & 110106 & 109668 & 5300.0 & 0.02 \\ 
  
\hline

\end{tabular}

	\caption{\label{tab:random_perf}Results for the random graphs.}
\end{table}

\clearpage

\subsection{Robustness}

For the following tables, the first column gives the name of the 
graph if the graph is from the Harwell-Boeing collection, of the value of 
$p$ is the graph is a random graph. The next two columns give the results 
for the heuristic \algo{}, with the median value of \cbs{} value over 30 
repetitions (\emph{median \cbs{}}) and the median absolute deviation (\emph{mad 
\cbs{}}) when $k$ repetitions are performed. The next four columns give the 
results for respectively $k=20$ and $k=30$. Finally the last column 
gives the minimum value of \cbs{} achieved for all repetitions.

\begin{table}[!ht]
	\centering
	\scriptsize
	\begin{tabular}{|l|rr|rr|rr|r|}
	\hline
\multicolumn{1}{|c|}{\textbf{Graph}} & 

\multicolumn{2}{c|}{$\bf{k=10}$} & 
\multicolumn{2}{c|}{$\bf{k=20}$} &  
\multicolumn{2}{c|}{$\bf{k=50}$} & \\ 
\hline	

\multicolumn{1}{|c|}{\emph{Name}} &

\multicolumn{1}{c}{\emph{median \cbs{}}} & 
\multicolumn{1}{c|}{\emph{mad \cbs{}}} &

\multicolumn{1}{c}{\emph{median \cbs{}}} & 
\multicolumn{1}{c|}{\emph{mad \cbs{}}} &

\multicolumn{1}{c}{\emph{median \cbs{}}} & 
\multicolumn{1}{c|}{\emph{mad \cbs{}}} &

\multicolumn{1}{c|}{\emph{min \cbs{}}}

\\ \hline

\makeatletter{}bcspwr01 & 102 & 1 & 101 & 1 & 100 & 1 & 99 \\ 
bcspwr02 & 157 & 1 & 156 & 0 & 155 & 0 & 154 \\ 
bcspwr03 & 784 & 17 & 767 & 10 & 756 & 6 & 722 \\ 
bcspwr04 & 5159 & 112 & 5083 & 200 & 4918 & 121 & 4543 \\ 
bcspwr05 & 5290 & 147 & 5166 & 105 & 5049 & 106 & 4469 \\ 
dwt59 & 280 & 4 & 277 & 4 & 274 & 3 & 261 \\ 
dwt72 & 184 & 6 & 183 & 5 & 177 & 4 & 169 \\ 
dwt87 & 1011 & 12 & 1000 & 8 & 995 & 5 & 979 \\ 
dwt162 & 2309 & 127 & 2228 & 75 & 2155 & 68 & 2025 \\ 
dwt193 & 28335 & 568 & 27340 & 478 & 27364 & 523 & 25693 \\ 
dwt221 & 5078 & 179 & 5087 & 147 & 4773 & 92 & 4475 \\ 
can24 & 192 & 2 & 192 & 2 & 190 & 0 & 190 \\ 
can61 & 1553 & 0 & 1553 & 0 & 1553 & 0 & 1553 \\ 
can62 & 226 & 5 & 220 & 5 & 215 & 4 & 203 \\ 
can73 & 949 & 17 & 940 & 14 & 914 & 16 & 888 \\ 
can96 & 2055 & 86 & 2051 & 91 & 1929 & 32 & 1832 \\ 
can144 & 11110 & 0 & 11110 & 0 & 11110 & 0 & 11110 \\ 
can187 & 4073 & 84 & 4031 & 135 & 3920 & 117 & 3544 \\ 
can229 & 8923 & 361 & 8530 & 260 & 8435 & 298 & 7574 \\ 
can268 & 30414 & 962 & 29929 & 969 & 28289 & 654 & 24826 \\

\hline
\end{tabular}

	\caption{\label{tab:harwell_sto}Results for the graph from the 
Harwell-Boeing collection}
\end{table}

\begin{table}[!ht]
	\centering
	\scriptsize
	\begin{tabular}{|l|rr|rr|rr|r|}
	\hline
\multicolumn{1}{|c|}{\textbf{Graph}} & \multicolumn{2}{c|}{$\bf{k=10}$} & 
\multicolumn{2}{c|}{$\bf{k=20}$} &  \multicolumn{2}{c|}{$\bf{k=50}$} & \\ 
\hline	

\multicolumn{1}{|c|}{\emph{p}} &
\multicolumn{1}{c}{\emph{median \cbs{}}} & \multicolumn{1}{c|}{\emph{mad 
\cbs{}}} &
\multicolumn{1}{c}{\emph{median \cbs{}}} & \multicolumn{1}{c|}{\emph{mad 
\cbs{}}} &
\multicolumn{1}{c}{\emph{median \cbs{}}} & \multicolumn{1}{c|}{\emph{mad 
\cbs{}}} &
\multicolumn{1}{|c|}{\emph{min \cbs{}}}

\\ \hline

\makeatletter{}0.9 & 11236 & 0 & 11236 & 0 & 11236 & 0 & 11236 \\ 
0.9 & 13169 & 0 & 13169 & 0 & 13169 & 0 & 13169 \\ 
0.9 & 12682 & 0 & 12682 & 0 & 12682 & 0 & 12682 \\ 
0.9 & 13013 & 0 & 13013 & 0 & 13013 & 0 & 13013 \\ 
0.9 & 14011 & 0 & 14011 & 0 & 14011 & 0 & 14011 \\ 
0.9 & 12295 & 0 & 12295 & 0 & 12295 & 0 & 12295 \\ 
0.9 & 12772 & 0 & 12772 & 0 & 12772 & 0 & 12772 \\ 
0.9 & 12439 & 0 & 12439 & 0 & 12439 & 0 & 12439 \\ 
0.9 & 11654 & 0 & 11654 & 0 & 11654 & 0 & 11654 \\ 
0.9 & 13036 & 0 & 13036 & 0 & 13036 & 0 & 13036 \\ 
0.9 & 37587 & 0 & 37587 & 0 & 37587 & 0 & 37587 \\ 
0.9 & 35784 & 0 & 35784 & 0 & 35784 & 0 & 35784 \\ 
0.9 & 38841 & 0 & 38841 & 0 & 38841 & 0 & 38841 \\ 
0.9 & 39000 & 0 & 39000 & 0 & 39000 & 0 & 39000 \\ 
0.9 & 39524 & 0 & 39524 & 0 & 39524 & 0 & 39524 \\ 
0.9 & 36670 & 0 & 36670 & 0 & 36670 & 0 & 36670 \\ 
0.9 & 36893 & 0 & 36893 & 0 & 36893 & 0 & 36893 \\ 
0.9 & 38280 & 0 & 38280 & 0 & 38280 & 0 & 38280 \\ 
0.9 & 38131 & 0 & 38131 & 0 & 38131 & 0 & 38131 \\ 
0.9 & 37388 & 0 & 37388 & 0 & 37388 & 0 & 37388 \\ 
0.9 & 60937 & 0 & 60937 & 0 & 60937 & 0 & 60937 \\ 
0.9 & 62238 & 0 & 62238 & 0 & 62238 & 0 & 62238 \\ 
0.9 & 61217 & 0 & 61217 & 0 & 61217 & 0 & 61217 \\ 
0.9 & 61659 & 0 & 61659 & 0 & 61659 & 0 & 61659 \\ 
0.9 & 62437 & 0 & 62437 & 0 & 62437 & 0 & 62437 \\ 
0.9 & 62970 & 0 & 62970 & 0 & 62970 & 0 & 62970 \\ 
0.9 & 63012 & 0 & 63012 & 0 & 63012 & 0 & 63012 \\ 
0.9 & 61796 & 0 & 61796 & 0 & 61796 & 0 & 61796 \\ 
0.9 & 61736 & 0 & 61736 & 0 & 61736 & 0 & 61736 \\ 
0.9 & 60957 & 0 & 60957 & 0 & 60957 & 0 & 60957 \\ 
0.9 & 87649 & 0 & 87649 & 0 & 87649 & 0 & 87649 \\ 
0.9 & 88077 & 0 & 88077 & 0 & 88077 & 0 & 88077 \\ 
0.9 & 86837 & 0 & 86837 & 0 & 86837 & 0 & 86837 \\ 
0.9 & 87006 & 0 & 87006 & 0 & 87006 & 0 & 87006 \\ 
0.9 & 87801 & 0 & 87801 & 0 & 87801 & 0 & 87801 \\ 
0.9 & 87372 & 0 & 87372 & 0 & 87372 & 0 & 87372 \\ 
0.9 & 88352 & 0 & 88352 & 0 & 88352 & 0 & 88352 \\ 
0.9 & 88322 & 0 & 88322 & 0 & 88322 & 0 & 88322 \\ 
0.9 & 87339 & 0 & 87339 & 0 & 87339 & 0 & 87339 \\ 
0.9 & 88014 & 0 & 88014 & 0 & 88014 & 0 & 88014 \\ 
0.9 & 111117 & 0 & 111117 & 0 & 111117 & 0 & 111117 \\ 
0.9 & 113297 & 0 & 113297 & 0 & 113297 & 0 & 113297 \\ 
0.9 & 113026 & 0 & 113026 & 0 & 113026 & 0 & 113026 \\ 
0.9 & 112714 & 0 & 112714 & 0 & 112714 & 0 & 112714 \\ 
0.9 & 113236 & 0 & 113236 & 0 & 113236 & 0 & 113236 \\ 
0.9 & 112747 & 0 & 112747 & 0 & 112747 & 0 & 112747 \\ 
0.9 & 111787 & 0 & 111787 & 0 & 111787 & 0 & 111787 \\ 
0.9 & 112802 & 0 & 112802 & 0 & 112802 & 0 & 112802 \\ 
0.9 & 112511 & 0 & 112511 & 0 & 112511 & 0 & 112511 \\ 
0.9 & 112342 & 0 & 112342 & 0 & 112342 & 0 & 112342 \\

\hline
\end{tabular}

	\caption{\label{tab:random_sto}Results for the random graphs.}
\end{table}

\begin{table}[!ht]
	\centering
	\scriptsize
	\begin{tabular}{|lr|rr|rr|rr|r|}
	\hline
\multicolumn{2}{|c|}{\textbf{Graph}}  & 
\multicolumn{2}{c|}{$\bf{k=10}$} & \multicolumn{2}{c|}{$\bf{k=20}$} &  
\multicolumn{2}{c|}{$\bf{k=50}$} & \\ \hline	

\multicolumn{1}{|c}{\emph{Name}} & \multicolumn{1}{c|}{\emph{Repetition}} &
\multicolumn{1}{c}{\emph{median \cbs{}}} & \multicolumn{1}{c|}{\emph{mad 
\cbs{}}} &
\multicolumn{1}{c}{\emph{median \cbs{}}} & \multicolumn{1}{c|}{\emph{mad 
\cbs{}}} &
\multicolumn{1}{c}{\emph{median \cbs{}}} & \multicolumn{1}{c|}{\emph{mad 
\cbs{}}} &
\multicolumn{1}{|c|}{\emph{min \cbs{}}}

\\ \hline

\makeatletter{}SF & 0 & 373 & 7 & 373 & 7 & 373 & 7 & 358 \\ 
SF & 1 & 374 & 2 & 374 & 2 & 374 & 2 & 366 \\ 
SF & 2 & 381 & 3 & 381 & 3 & 381 & 3 & 366 \\ 
SF & 3 & 315 & 3 & 315 & 3 & 315 & 3 & 305 \\ 
SF & 4 & 391 & 1 & 391 & 1 & 391 & 1 & 388 \\ 
SF & 5 & 306 & 3 & 306 & 3 & 306 & 3 & 296 \\ 
SF & 6 & 351 & 1 & 351 & 1 & 351 & 1 & 345 \\ 
SF & 7 & 379 & 1 & 379 & 1 & 379 & 1 & 374 \\ 
SF & 8 & 401 & 2 & 401 & 2 & 401 & 2 & 394 \\ 
SF & 9 & 354 & 3 & 354 & 3 & 354 & 3 & 347 \\ 
SW & 0 & 752 & 0 & 752 & 0 & 752 & 0 & 752 \\ 
SW & 1 & 937 & 0 & 937 & 0 & 937 & 0 & 937 \\ 
SW & 2 & 662 & 0 & 662 & 0 & 662 & 0 & 662 \\ 
SW & 3 & 977 & 0 & 977 & 0 & 977 & 0 & 977 \\ 
SW & 4 & 571 & 0 & 571 & 0 & 571 & 0 & 571 \\ 
SW & 5 & 741 & 0 & 741 & 0 & 741 & 0 & 741 \\ 
SW & 6 & 742 & 0 & 742 & 0 & 742 & 0 & 742 \\ 
SW & 7 & 891 & 0 & 891 & 0 & 891 & 0 & 891 \\ 
SW & 8 & 563 & 0 & 563 & 0 & 563 & 0 & 563 \\ 
SW & 9 & 1035 & 0 & 1035 & 0 & 1035 & 0 & 1035 \\ 
COM & 0 & 26338 & 20 & 26338 & 20 & 26338 & 20 & 26292 \\ 
COM & 1 & 27352 & 37 & 27352 & 37 & 27352 & 37 & 27315 \\ 
COM & 2 & 27880 & 10 & 27880 & 10 & 27880 & 10 & 27831 \\ 
COM & 3 & 31573 & 0 & 31573 & 0 & 31573 & 0 & 31573 \\ 
COM & 4 & 31599 & 49 & 31599 & 49 & 31599 & 49 & 31515 \\ 
COM & 5 & 27067 & 49 & 27067 & 49 & 27067 & 49 & 26822 \\ 
COM & 6 & 29546 & 22 & 29546 & 22 & 29546 & 22 & 29400 \\ 
COM & 7 & 28102 & 136 & 28102 & 136 & 28102 & 136 & 27890 \\ 
COM & 8 & 27571 & 22 & 27571 & 22 & 27571 & 22 & 27549 \\ 
COM & 9 & 26943 & 0 & 26943 & 0 & 26943 & 0 & 26943 \\

\hline
\end{tabular}

	\caption{\label{tab:complex_sto}Results for the complex networks.}
\end{table}

\clearpage

\bibliographystyle{plain}
\bibliography{bib}

\end{document}